\documentclass[journal,onecolumn]{IEEEtran}

\pdfoutput=1

\usepackage[tbtags]{amsmath}
\usepackage{amsthm,amssymb}
\usepackage{graphicx}
\usepackage{subfig}
\usepackage{cite}
\usepackage{xspace}
\usepackage{bm}

\theoremstyle{plain}
\newtheorem{theorem}{Theorem}
\newtheorem{lemma}{Lemma}
\newtheorem{prop}{Proposition}
\newtheorem{remark}{Remark}
\newtheorem{defi}{Definition}

\newcommand{\dfn}{\triangleq}
\newcommand{\given}{\mid}

\newcommand{\thetaSpace}{\Theta}

\DeclareMathOperator{\sign}{sgn}
\DeclareMathOperator{\trace}{tr}
\DeclareMathOperator{\expect}{\mathbb{E}}
\DeclareMathOperator{\gradient}{\nabla}
\DeclareMathOperator*{\argument}{\arg}

\newcommand*\interior[1]{#1^{\mathsf{o}}}
\newcommand{\VorEdge}{\partial\mathcal{V}}
\newcommand{\Vor}[1]{\mathcal{V}_{#1}}

\DeclareMathOperator{\EZvorOne}{\mathbb{E}_z^{\mathcal{V}_{1}}}
\DeclareMathOperator{\EZvorTwo}{\mathbb{E}_z^{\mathcal{V}_{2}}}

\DeclareMathOperator{\EZvorI}{\mathbb{E}_z^{\mathcal{V}_{i}}}
\DeclareMathOperator{\EZvorOneZeta}{\mathbb{E}_z^{\mathcal{V}_{1}(\hat{\bm{\zeta}})}}
\DeclareMathOperator{\EZvorTwoZeta}{\mathbb{E}_z^{\mathcal{V}_{2}(\hat{\bm{\zeta}})}}

\DeclareMathOperator{\PZvor}{\mathbb{P}_z}

\DeclareMathOperator{\PZvorOne}{\mathbb{P}_z(\mathcal{V}_1)}
\DeclareMathOperator{\PZvorTwo}{\mathbb{P}_z(\mathcal{V}_2)}

\newcommand{\Estimator}[1]{\bm{\hat\theta}^{(#1)}}
\newcommand{\transf}{\ensuremath{\bm{\tau}}\xspace}

\let\abs=\envert

\let\norm=\enVert

\newcommand{\paren}[1]{\left(#1\right)}
\newcommand{\gaus}[2]{\mathcal{N}\paren{#1,#2}}

\newcommand{\wHier}{w_\Hier}
\newcommand{\Heter}{\text{HT}}
\newcommand{\Hier}{\text{HI}}
\newcommand{\EstimSpace}{L^2(\mathcal{\bm{Z}})}
\newcommand{\EstimSpaceRealization}{\thetaSpace}  %%  YO 19.7.22
\newcommand{\augmentedSpaceRealization}{\EstimSpaceRealization\times\EstimSpaceRealization}
\newcommand{\transp}{^T}
\newcommand{\givenMeas}{\given \bm{Z}}
\newcommand{\JZ}{J_Z}
\newcommand{\J}{J}
\newcommand{\JMS}{J_{\text{MS}}}
\newcommand{\Sset}{\mathcal{S}_{\hat{\bm{\theta}}}}
\newcommand{\thetaMS}{\hat{\bm{\theta}}_{\text{MS}}}

\newcommand{\PositiveReals}{\mathbb{R}_{>0}}

\newcommand{\JZrest}{\JZ\restriction_\mathcal{T}}

\renewcommand{\restriction}{\mathord{\upharpoonright}}

\usepackage{hyperref}

\DeclareGraphicsExtensions{.pdf,.eps}

\begin{document}
\title{Double-Opportunity Estimation via Altruism}

\author{Nitai Stein and Yaakov Oshman,~\IEEEmembership{Fellow,~IEEE}%
  \thanks{The authors are with the Department of Aerospace
    Engineering, Technion---Israel Institute of Technology, Haifa
    32000, Israel. (E-mail: nitais@alumni.technion.ac.il;
    yaakov.oshman@technion.ac.il).}}

\maketitle

\begin{abstract}
  Based on the notion of altruism, we present an approach to
  cooperative parameter estimation in a system comprising two
  information-sharing agents. The underlying assumption is that the
  overall two-agent scheme can reach desired performance level even if
  only one of the agents performs satisfactorily, hence there exist
  two independent opportunities to estimate. The notion of altruism
  motivates a definition of cooperative estimation optimality that
  generalizes the common definition of minimum mean squared error
  optimality.  Fundamental equations are derived for two types of
  altruistic cooperative estimation problems, corresponding to
  heterarchical and hierarchical setups. Although these equations are,
  generally, hard to solve, their solution in the Gaussian case is
  straightforward and only entails the computation of the largest
  eigenvalue of the conditional covariance matrix and its
  corresponding eigenvector.  Moreover, in the Gaussian case the
  performance improvement of the two altruistic cooperative estimation
  techniques over the conventional (egoistic) estimation approach is
  shown to depend on the problem's dimensionality and statistical
  distribution. In particular, the performance improvement grows with
  the dispersion of the spectrum of the conditional covariance matrix,
  rendering the presented estimation approach especially appealing in
  ill-conditioned problems.  The validity of the solution in the
  Gaussian case is illustrated numerically.
\end{abstract}
\section{Introduction}\label{sec:intro}
Complex missions often involve a number of systems, or agents,
operating together as a team, to promote flexibility and robustness
and to improve overall performance. In such teamwork, it is highly
advantageous for the member agents constituting the team to be capable
of sharing information among themselves.  The need for improving
teamwork, as well as the capability to share information among team
members, have led to accelerated advances in the research of
cooperative estimation. Current cooperative estimation algorithms
differ in the way they handle shared information among the nodes of
the distributed system, but, in the end, in most cases a \emph{common
  team estimate} is computed, which is then used by the entire team.

In contradistinction to the prevailing team estimation concept, this
paper introduces a nonorthodox paradigm in cooperative parameter
estimation, whereby local estimates, that are generated by separate
(but information-sharing) agents, are --- by design --- not identical,
and they are not merged to form a unified estimate. Possibly even
sub-optimal in the standard, minimum mean-squared error (MMSE) sense,
these local estimates are designed, instead, to minimize together a
single, global, system-oriented cost.

To illustrate this cooperative estimation paradigm and its
applicability, consider the following scenario.  An attacker (say, an
aircraft equipped with high-precision guided missiles) is tasked with
destroying a static high-valued target whose precise location is not
deterministically known (say, a well-hidden rocket launcher). We
assume that the attacker has two launch opportunities, i.e., it can
fire two missiles at the target. We also assume that the attacker
knows (either through intelligence sources or via its self-acquired
measurements) the target location's distribution function. Clearly, as
long as it is eventually successful in destroying the target, the
attacker is completely indifferent as to how its mission is actually
accomplished, that is, which of its two launch opportunities is
successful in hitting the target. So, assuming that the attacker's two
missiles are identical, and that they rely on the same positional
information, how should they be aimed? If they are aimed at a common
point, say, the mean of the target location's distribution function
--- which would be an optimal aiming point from an estimation
perspective --- this would amount to wasting one missile (assuming
that both missiles operate flawlessly). Indeed, as will be shown later
in this paper, the attacker will benefit from dispersing its two
shooting opportunities, aiming them at different (possibly suboptimal
from an individual missile's estimation standpoint) locations, in
order to maximize its overall success probability.

The key assumption underlying the introduced cooperative estimation
paradigm is that the system encompasses an inherent redundancy, which,
in the scenario illustrated above, is embodied in the attacker having
two missiles.  Frequently called for to improve the probability of
success in critically important missions, redundancy is implemented by
using more than the minimal number of agents required to perform a
certain, global, task.  For example, in the theatre ballistic missile
defense world, several (identical) defending interceptors may be
launched at a single oncoming threat, if the defended target is so
valuable that it must be protected at all
cost~\cite{shaferman_oshman_JGCD:2016}.  While system redundancy
obviously contributes to immunity against local, subsystem failures,
it is proposed herein to exploit it in a seemingly unrelated
manner. Thus, we note that when the system comprises several identical
and information-sharing {agents}, it would be advantageous to put
aside the performance of each individual estimator {(each agent's
  estimate)} and, instead, to focus on enhancing the estimation
performance of the entire system in some global sense.  Such design
philosophy suggests an altruism-based cooperation, in which each agent
forgoes its own (egoistic) estimation performance in order to maximize
a global estimation performance measure.  A well-known term in nature
and in sociology, as well as in game theory, altruism means that an
individual sacrifices itself for the greater good of its species, or
in favor of other individuals, in order to improve the chances of its
species to thrive.  The use of altruism leads conceptually to a
min-min game, that is, a game where all players (belonging to one
side) cooperate such that there is one global goal to achieve, and the
optimizer aspires to minimize a cost function based on the minimal
cost among all players. The underlying notion is that the success of
the global mission is determined only by the performance of the most
successful individual among all {agents}. The interested reader is
referred to~\cite{hayoun_EuroGNC:2015,Russian_IFAC} for
perfect-information examples of such min-min games in the field of
missile guidance, where, however, the estimation problem was not
addressed.

Two approaches for cooperative parameter estimation based on the
notion of altruism are proposed herein.  Called {heterarchical
  altruistic cooperative estimation}, the first approach considers two
equally-ranked {agents} that take into account the action of each
other, and calculate their estimates fully altruistically such that
neither of the two is {necessarily} optimal (in the conventional
sense). Thus, both {agents} sacrifice their own estimation performance
in order to maximize the global mission's probability of success.
Termed {hierarchical altruistic cooperative estimation}, the second
approach is more conservative, in that it assumes that one of the
{agents} operates egoistically, as if it were the only estimator
present, thus minimizing the conventional mean-squared error (MSE)
criterion.  The second {agent} in this approach takes into account the
action of the first one, and maximizes the success of the global
mission, given the (egoistic) estimate of the first {agent}. Here the
term hierarchy alludes to the fact that the first {agent} works as if
it were ranked higher than the other, and tries to accomplish the
mission on its own.  Comparing the two approaches in performance, the
heterarchical one is superior to the hierarchical one, as the former
is globally optimal, whereas the latter results from a constrained
optimization. However, in some cases, design conservatism may dictate
the use of the hierarchical approach.

Returning to the scenario of the attacker and its two shooting
  attempts, assume now that the target's location is Gaussian
distributed on a line. We will show in this paper that an attacker
using the globally-optimal heterarchical approach would shift both of
its missiles' aiming points to the sides of the {target}'s expected
location (the mode of the distribution), thus maximizing the overall
probability of at least one of its missiles striking sufficiently
close to the target.  Alternatively, the more conservative {attacker,
  that implements the} hierarchical approach, would aim one of its
missiles at the expected location of the {target} (maximizing this
particular missile's probability of hitting the {target}), and would
shift the aiming point of the second missile aside (reasonably, the
second missile's aiming point's shift from the expected target
location would be larger than the respective shifts of both
heterarchical {attempts}).

The main contributions of this paper are the following:
\begin{enumerate}
\item {We generalize} the standard MMSE parameter estimation
  problem to the realm of cooperative estimation, in cases involving
  two separately-operating but information-sharing {agents}.
\item We introduce the concept of altruistic cooperative estimation,
  and, within this concept, we introduce two altruistic estimation
  approaches (heterarchical and hierarchical estimation), that
  outperform the standard, egoistic approach {in the {sense}
    considered herein}.
\item We prove the existence of a global solution for each of the two
  altruistic estimation approaches.
\item In the Gaussian case we provide closed-form, analytical
  solutions to the two altruistic estimation problems, along with a
  complete analysis of the solutions' potential performance
  advantages.
\end{enumerate}

The remainder of this paper is organized as follows.  The two
altruistic cooperative estimation problems are defined in
Section~\ref{sec:probdef}.  In Section~\ref{sec:est_general} we derive
necessary conditions for estimators corresponding to both problems,
and prove the existence of optimal estimators satisfying these
conditions.  In Section~\ref{sec:est_gauss} we address the Gaussian
case in detail.  Concluding remarks are offered in the last section.
Some technical derivations and proofs are deferred to Appendices.

\section{Problem Formulation}\label{sec:probdef}
Consider a random parameter vector $\bm{\theta}$ defined on the
probability space $(\thetaSpace,\mathcal{F},P)$, where
$\thetaSpace \subseteq \mathbb{R}^n$ is the continuous sample space,
$\mathcal{F}$ is the set of events ($\sigma$-algebra) on
$\thetaSpace$, and $P$ is a probability
measure.  The problem is to {find {two} estimates of} $\bm{\theta}$
based on the random vector of measurements $\bm{Z}$, which {are}
(possibly nonlinear) functions of $\bm{\theta}$.  The mapping $\bm{Z}$
induces the sample space $\mathcal{\bm{Z}} \subseteq \mathbb{R}^m$
(with an appropriate $\sigma$-algebra). Both $\thetaSpace$ and
$\mathcal{\bm{Z}}$ are Hilbert spaces, equipped with the $2$-norm
induced by the dot-product. {For later purpose,} we assume that the (known) joint
distribution of $\bm{\theta}$ and $\bm{Z}$ has finite first two
moments.

{We consider a scenario where} the system tasked with the estimation
problem comprises two {agents}, each of which {yields a local estimate
  of} the parameter vector $\bm{\theta}$ based on the {shared}
measurements $\bm{Z}$. {The system does not merge the two local
  estimates to a final, single estimate; rather, its overall
  performance {results}, in some manner, from the {joint}
  performance of the two estimators.}  Context-depending, we will use
the notation $\hat{\bm{\theta}}^{(1)}$ and $\hat{\bm{\theta}}^{(2)}$
to denote both the estimators and the estimates (generated by these
estimators) of the two {agents}, respectively.

{A cost function that reflects the idea of altruistic estimation
  is the following:}
\begin{equation}
\label{eq:CostFunctionDefinition}
\J{(\hat{\bm{\theta}}^{(1)},\hat{\bm{\theta}}^{(2)})}
\dfn
\expect (
\norm{\hat{\bm{\theta}}^{(1)}-\bm{\theta}}^2
\wedge
\norm{\hat{\bm{\theta}}^{(2)}-\bm{\theta}}^2
)
\end{equation}
where $\expect$ is the expectation operator, and
$a \wedge b \dfn \min(a,b)$ for some $a,b \in \mathbb{R}$.

The underlying premise of this work is that the global mission is
accomplished even if only one of the {agents} provides an
MSE-acceptable estimate. Thus, the overall system performance is
determined by the performance of the better {agent} among the two.

 \begin{remark}
   {Setting the two estimators to be identical in
   (\ref{eq:CostFunctionDefinition}) reduces it to the standard
   {MMSE} cost, manifesting the fact that
   the problem defined here is an extension of the MMSE estimation
   problem to the realm of altruistic cooperative estimation.}
 \end{remark}
 \begin{remark}
   Somewhat resembling the optimal sub-pattern assignment (OSPA)
   metric of \cite{OSPA,birds}, the cost
   (\ref{eq:CostFunctionDefinition}) as defined here is radically
   different due to the difference between the meanings of both
   problems. In the OSPA case, the best targets-to-estimates
   combination is chosen, based on the premise that targets are
   unlabeled, so that the problem is how to optimally ``throw two
   stones at two indistinguishable birds'', aiming at hitting both. In
   contradistinction, in the present work the problem is ``how to
   throw two stones at a single bird'', while maximizing the
   probability that at least one (unlabeled) stone hits its target.
\end{remark}
We define two altruistic estimation problems. In the heterarchical
problem, the estimators $\hat{\bm{\theta}}^{(1)}_\Heter$ and
$\hat{\bm{\theta}}^{(2)}_\Heter$ solve the global minimization problem
\begin{equation}
\label{eq:Heter:Definition}
\min_{\bm{\hat{\theta}}^{(1)},\bm{\hat{\theta}}^{(2)} \in \EstimSpace}
\J{(\hat{\bm{\theta}}^{(1)},\hat{\bm{\theta}}^{(2)})}
\end{equation}
where $\EstimSpace$ is the space of all square Lebesgue-integrable
(measurable) functions of the measurements.

A constrained version of the heterarchical problem, the hierarchical
problem sets one of the estimators, $\bm{\hat{\theta}}^{(1)}_\Hier$,
identical to the minimum mean squared error estimator (MMSEE),
$\thetaMS$. The second hierarchical altruistic estimator,
$\hat{\bm{\theta}}^{(2)}_\Hier$, solves the constrained minimization
problem
\begin{equation}
\label{eq:Hier:Definition}
\min_{\bm{\hat{\theta}}^{(2)} \in \EstimSpace}
J{(\hat{\bm{\theta}}^{(1)},\hat{\bm{\theta}}^{(2)})} \quad  
\text{such that  } \bm{\hat{\theta}}^{(1)} = \thetaMS.
\end{equation}

  Problems \eqref{eq:Heter:Definition} and \eqref{eq:Hier:Definition}
  are closely related to some well-studied optimization problems
  appearing in the Voronoi literature
  \cite{Voronoi_book,Voronoi_old,CVT}, generally called \emph{facility
    serviceability problems} \cite{Voronoi_old}. In these problems,
  there exist some points called \emph{facilities} (or, \emph{Voronoi
    generators}), that are said to supply some necessary resource. The
  optimization task is to localize these facilities inside a populated
  region, such that they generate a Voronoi tessellation which is
  optimal in some sense. The simplest example is the problem of
  \emph{public mail box localization}: given a city and its population
  distribution, the problem is to position a certain number of public
  mail boxes, assuming that every citizen in the city uses the closest
  public mail box. These problems appear in many scientific domains
  \cite{CVT}, such as data compression (e.g, in the image processing
  world), quantization, and distortion problems \cite{motion_coord} in
  the signal compression world.  However, to the best of the authors'
  knowledge, no closed-form solutions have been presented, perhaps
  because the literature focuses mainly on problems with many
  facilities, that require efficient numerical solutions, such as
  Lloyd's algorithm \cite{CVT,Voronoi_old,Voronoi_book}.

%%% Local Variables:
%%% mode: latex
%%% TeX-master: "main"
%%% End:

\section{Estimator Derivation}\label{sec:est_general}
Applying the smoothing theorem to the cost function
(\ref{eq:CostFunctionDefinition}) yields
\begin{equation}\label{eq: smoothing theorem}
\J =  \expect [  \expect ( 
\norm{\hat{\bm{\theta}}^{(1)}-\bm{\theta}}^2 \wedge
\norm{\hat{\bm{\theta}}^{(2)}-\bm{\theta}}^2
\givenMeas )]
 = \expect \JZ
\end{equation}
where
\begin{equation}\label{eq: definition of \JZ}
\JZ \dfn  \expect ( 
\norm{\hat{\bm{\theta}}^{(1)}-\bm{\theta}}^2 \wedge
\norm{\hat{\bm{\theta}}^{(2)}-\bm{\theta}}^2
\givenMeas ).
\end{equation}
Since the outer expectation in (\ref{eq: smoothing theorem}) does not
depend on the choice of the estimators, the global minimizing
arguments for $\J$ are identical to those of $\JZ$. We, thus, proceed
with minimizing $\JZ$.

Consider the function $a \wedge b$ for some $a,b\in\mathbb{R}$.
Clearly, in the region $a>b$, $a \wedge b = b$, so that the function
is not affected by the value of $a$.  Analogously, the space
$\thetaSpace$ can be divided into two subspaces, in each of which $\JZ$
is affected by only one of the two {estimates} -- the one closer to
any value of $\bm \theta$ in this subspace. This observation naturally
calls to mind the notion of Voronoi regions~\cite{CVT}, giving rise to
the following definition.

\begin{defi}[{Estimates'} Voronoi regions]
\label{defi: governance}
The Voronoi region of $\hat{\bm{\theta}}^{(1)}$,
denoted $\mathcal{V}_1$, is a set in $\thetaSpace$ such that:
\begin{equation} \label{eq:defi:D1}
\norm{\hat{\bm{\theta}}^{(1)}-\bm{\theta}}
<
\norm{\hat{\bm{\theta}}^{(2)}-\bm{\theta}}
\quad\forall \bm{\theta} \in \mathcal{V}_1.
\end{equation}
Analogously, $\mathcal{V}_2$ is defined to satisfy
\begin{equation} \label{eq:defi:D2}
\norm{\hat{\bm{\theta}}^{(1)}-\bm{\theta}}
>
\norm{\hat{\bm{\theta}}^{(2)}-\bm{\theta}}
\quad\forall \bm{\theta} \in \mathcal{V}_2.
\end{equation}
The boundary separating both Voronoi regions (the Voronoi edge),
denoted as $\VorEdge$, is defined to satisfy:
\begin{equation} \label{eq:defi:Dbar}
\norm{\hat{\bm{\theta}}^{(1)}-\bm{\theta}}
=
\norm{\hat{\bm{\theta}}^{(2)}-\bm{\theta}}
\quad\forall \bm{\theta} \in \VorEdge.
\end{equation}
\end{defi}
Notice that the two {estimates}, $\Estimator{1}$ and $\Estimator{2}$,
play the part of Voronoi generators, and that the regions $\Vor{1}$
and $\Vor{2}$ along with $\VorEdge$ constitute a Voronoi
tessellation\cite{CVT}.  The optimization problems
(\ref{eq:Heter:Definition}) and (\ref{eq:Hier:Definition}) are Voronoi
optimization problems, for, using the law of total probability, we can
express $\JZ$ as
 \begin{equation}
 \label{eq:CostFunctionVoronoiForm}
 \JZ = \sum_{i=1}^{2} \expect (
   \norm{\hat{\bm{\theta}}^{(i)}-\bm{\theta}}^2
   \givenMeas,\bm{\theta} \in \mathcal{V}_i 
 )
 \Pr \left(\bm{\theta} \in \mathcal{V}_i \givenMeas \right)
 \end{equation}
 where we have used the fact that $\VorEdge \subset \thetaSpace$ has
 measure zero, as $ \dim \VorEdge = \dim \thetaSpace - 1$ (because
 $\VorEdge$ satisfies a constraint in $\thetaSpace$).

 The heterarchical problem stated in (\ref{eq:Heter:Definition}) is
 the classical Voronoi facility serviceability
 problem\cite{Voronoi_old}. For given measurements, the estimates are
 the facilities, located in $\mathbb{R}^n$; the probability
 distribution of $\bm{\theta}$ serves as the population distribution,
 and each individual (random realization of $\bm{\theta}$) is
 associated with the estimate closest to it.
 
 The hierarchical problem stated in (\ref{eq:Hier:Definition}) is a
 special case of the problem stated in \cite[Section
 9.2.4]{Voronoi_book}. To see this, set the number of facilities to
 two (represented by the two {estimates}), the number of ranks to two
 (one egoistic {estimate} and one altruistic), and the consumption rate
 for each supply to half, such that both are equally consumed. In that
 case, it is obvious that the location of the higher ranked facility
 should be the MMSEE, since it is the only facility supplying this
 service to the entire population. However, this higher-ranked
 facility, whose location is already set, supplies also the second
 service which the other facility supplies as well. Hence, it is
 obvious that the optimizer should locate the lower-ranked facility
 according to the global mission of serviceability, taking into
 consideration the location of the higher-ranked facility.

 For these kinds of problems, \cite{CVT} proves that the optimal
 solutions lead to {centroidal Voronoi tessellations} (CVT), in which
 the facilities are the centroids of their corresponding Voronoi
 regions. Requiring the domain in which the problem is defined to be a
 compact subset of $\mathbb{R}^n$, \cite[p. 651-652]{CVT} proves the
 existence of a globally-optimal solution, that consists of a set of
 non-identical facilities.  The localization problems addressed in the
 literature are commonly solved numerically, perhaps because most of
 them involve a large number of
 facilities~\cite{CVT,Voronoi_old,Voronoi_book}.  Related to
   what would be called the 2-facilities problem in the Voronoi
   literature, our work extends already known results by proving the
   existence of a globally-optimal solution in not necessarily compact
   domains, and by providing a {closed-form} solution in the Gaussian
   case.
\subsection{Preliminary Calculations}
Because we {allow} unconstrained {estimates}, the minimizers of $\JZ$
are those for which either the gradient vector vanishes, or the cost
function is not differentiable. {The function $a \wedge b$ for some
  $a,b\in\mathbb{R}$ is differentiable everywhere with respect to both
  $a$ and $b$, except at $a=b$. Similarly, $\JZ$ is not differentiable
  {with respect to either of the two estimates} only in the trivial
  case $\hat{\bm{\theta}}^{(1)} = \hat{\bm{\theta}}^{(2)}$, which is
  not of interest here (one can always do better by dispersing the two
  {estimates}; see Lemma 2 in
    Appendix~\ref{sec:app_proof_opt_heter}). We, thus, seek for
  optimal solutions rendering
  $\hat{\bm{\theta}}^{(1)} \neq \hat{\bm{\theta}}^{(2)}$.  Since, for
  such solutions, $\JZ$ is differentiable, we will derive the
  necessary conditions by setting its gradient to zero. Notice
    that because $\JZ$ is quadratic (assuming
    $\hat{\bm{\theta}}^{(1)} \neq \hat{\bm{\theta}}^{(2)}$), its
    extremum is necessarily a minimum.

To compute the gradient
of $\JZ$ we first rewrite \eqref{eq:CostFunctionVoronoiForm} as

\begin{align}\label{eq:Lemma_AAAomprised}
  \JZ(\Estimator{1},\Estimator{2})     = \EZvorOne (
  \norm{\hat{\bm{\theta}}^{(1)}-\bm{\theta}}^2
  ) \PZvorOne
  + \EZvorTwo ( 
  \norm{\hat{\bm{\theta}}^{(2)}-\bm{\theta}}^2
  ) \PZvorTwo
\end{align}
where the (probabilistic) measure of Voronoi region $i$ and the local
expectation operator associated with that region are defined,
respectively, as
\begin{gather}\label{eq: def of EDi and PDi}
  \PZvor(\mathcal{V}_i) \dfn \Pr \left(\bm{\theta} \in \mathcal{V}_i \givenMeas \right), \quad i=1,2\\
\intertext{and}
\label{eq: EVi def}
\EZvorI(\cdot) \dfn \expect (\cdot \givenMeas, 
\bm{\theta} \in \mathcal{V}_i), \quad i=1,2.
\end{gather}
Let $\delta(\cdot)$ denote an infinitesimal perturbation of
  $(\cdot)$.  Arbitrarily perturbing the first {estimate} to
$\Estimator{1} + \delta\Estimator{1}$ while keeping the second
{estimate} intact, results in a corresponding infinitesimal change in
the Voronoi tessellation.  In turn, this results in a perturbation in
the cost,
\begin{align}
  \label{eq:cost_change}
  \delta\JZ
      =
    \delta( \EZvorOne (
  \norm{\hat{\bm{\theta}}^{(1)}-\bm{\theta}}^2
  ) \PZvorOne )  
 + \delta ( \EZvorTwo ( 
  \norm{\hat{\bm{\theta}}^{(2)}-\bm{\theta}}^2
  ) \PZvorTwo ).
\end{align}
The infinitesimal change in the tessellation affects both terms
  on the RHS of \eqref{eq:cost_change}, but it does so in an
  antisymmetric manner, as the change in $\VorEdge$, the boundary
  separating the two Voronoi regions, induces oppositely signed
  changes of both terms. Letting $\delta\Estimator{1} \to {\bm{0}}$
  nullifies the change in the tessellation, and its total effect on
  the perturbation of the cost. The other effect contributing to the
  perturbation of the cost is the change (integrated over all points
  $\bm{\theta}\in\Vor{1}$) in the norm
  $\norm{\hat{\bm{\theta}}^{(1)}-\bm{\theta}}$ due to the change in
  $\Estimator{1}$. Thus,
\begin{equation}
  \label{eq:cost_change_2}
  \delta\JZ = 2 \EZvorOne (\hat{\bm{\theta}}^{(1)}-\bm{\theta})^T
  \delta\Estimator{1} \PZvorOne.
\end{equation}
By symmetry, \eqref{eq:cost_change_2} yields the gradients of $\JZ$ as
\begin{equation}\label{eq: grad \JZ}
  \gradient_{\hat{\bm{\theta}}^{(i)}} \JZ
  =  2 \PZvor(\mathcal{V}_i) \EZvorI
  (\hat{\bm{\theta}}^{(i)}-\bm{\theta}),  \quad   i=1,2.
\end{equation}
\begin{remark}
  Addressing the multi-dimensional case, \eqref{eq: grad \JZ} was also
  derived in \cite{Voronoi_book,Voronoi_old}, {albeit} in a
  deterministic setting.
\end{remark}

Next we show that there exists a rotation transformation, that, when
applied to the parameter space $\Theta$, maps the Voronoi regions of
both {estimates} to one-dimensional, half-infinite intervals. This
transformation will facilitate the ensuing derivation of the
altruistic estimators. 
Defining
\begin{equation}\label{eq:Lemma_Deltas_definition}
\Delta \hat{\bm{\theta}} \dfn
 \hat{\bm{\theta}}^{(2)} - \hat{\bm{\theta}}^{(1)}
\end{equation}
the Voronoi edge equation, (\ref{eq:defi:Dbar}), can be written as
\begin{equation} \label{eq: Lemma_Dbar}
\langle \bm{\theta}-\frac{\hat{\bm{\theta}}^{(1)}+\hat{\bm{\theta}}^{(2)}}{2}
,
\Delta\hat{\bm{\theta}} \rangle
=
0
\end{equation}
where $\langle \cdot , \cdot \rangle$ stands for the inner product.
Similarly, the definitions (\ref{eq:defi:D1}) and (\ref{eq:defi:D2})
of the {estimates'} Voronoi regions can be written {(for any
  $\hat{\bm{\theta}}^{(1)}$ and $\hat{\bm{\theta}}^{(2)}$)} as
\begin{gather} \label{eq:voronoi_1}
\langle \bm{\theta}-\frac{\hat{\bm{\theta}}^{(1)}+\hat{\bm{\theta}}^{(2)}}{2},
\Delta\hat{\bm{\theta}} \rangle
< 0\\
\intertext{and}
\langle \bm{\theta}-\frac{\hat{\bm{\theta}}^{(1)}+\hat{\bm{\theta}}^{(2)}}{2},
\Delta\hat{\bm{\theta}} \rangle
> 0
\end{gather}
respectively.  Equation \eqref{eq: Lemma_Dbar} means that the boundary
$\VorEdge$ is an $(n-1)$-dimensional plane orthogonal to
$\Delta\hat{\bm{\theta}}$, that contains the point
$\frac{\hat{\bm{\theta}}^{(1)}+\hat{\bm{\theta}}^{(2)}}{2}$, the
mid-point between the two {estimates}. The two Voronoi regions are
located on opposite sides of the boundary.

Now let \transf be an $n \times n$ proper orthogonal matrix having
$\Delta \hat{\bm{\theta}}\transp/\norm{\Delta \hat{\bm{\theta}}}$ as
its first row ($\norm{\Delta \hat{\bm{\theta}}}$ cannot vanish
because, as explained earlier, we disregard the case
$\hat{\bm{\theta}}^{(1)} = \hat{\bm{\theta}}^{(2)}$).  Define
\begin{align}
  \bm{u} & \dfn \transf \bm{\theta}   \label{eq:def_u}\\  
  \hat{\bm{u}}^{(i)} & \dfn \transf \hat{\bm{\theta}}^{(i)}, \quad i=1, 2\label{eq:uhatdef}     \\
  \Delta\hat{\bm{u}} & \dfn  \hat{\bm{u}}^{(2)} - \hat{\bm{u}}^{(1)}\label{eq:deltauhatdef}.
\end{align}
Rotating the standard basis of the space $\Theta$ using the
transformation \transf, let $\bm{e_{u_1}}$ be a unit vector along the first
basis vector of the rotated space.  Using
\eqref{eq:Lemma_Deltas_definition} and \eqref{eq:uhatdef} in
\eqref{eq:deltauhatdef}, and recalling the special construction of the
orthogonal matrix \transf, yields
\begin{equation}\label{eq:Delta_u_hat definition}
  \Delta\hat{\bm{u}} = \norm{\Delta \hat{\bm{\theta}}} \bm{e_{u_1}}.
\end{equation}
Because $\bm{e_{u_1}}$ is collinear with $\Delta\hat{\bm{u}}$, the
vector connecting both transformed {estimates}, we call it the
{solution-axis}.

Using $\transf^T \transf = \bm{I}$ in \eqref{eq:voronoi_1} along with
\eqref{eq:Delta_u_hat definition} and the definitions \eqref{eq:def_u}
and \eqref{eq:uhatdef}, the Voronoi region of the first {estimate}
can be expressed as
\begin{equation}\label{eq:Lemma_D1 in eu basis}
\mathcal{V}_1 = \{ \bm{u} \in \thetaSpace:
\langle \bm{u}-\frac{\hat{\bm{u}}^{(1)}+\hat{\bm{u}}^{(2)}}{2},
\norm{\Delta \hat{\bm{\theta}}} \bm{e_{u_1}} \rangle <0 \}.
\end{equation}
Let $\hat{u}^m$ be the projection of the midpoint between the two
{estimates} in the transformed space on $\bm{e_{u_1}}$, that is
\begin{equation}\label{eq:umdef}
  \hat{u}^m \dfn \frac{1}{2}(\hat{\bm{u}}^{(1)} + \hat{\bm{u}}^{(2)})^T \bm{e_{u_1}}.
\end{equation}
Using \eqref{eq:umdef} in \eqref{eq:Lemma_D1 in eu basis} yields
\begin{equation}\label{eq: D1 is 1D in u}
\mathcal{V}_1 = \{\bm{u} \in \thetaSpace : u_1 < \hat{u}^m \}
\end{equation}
where $u_1$ denotes the first component of the vector $\bm{u}$ (its
projection along $\bm{e_{u_1}}$). Stated in words, in the transformed
space, $\mathcal{V}_1$ is the half-infinite open interval
$\left(-\infty,\hat{u}^m\right)$ along the solution axis. Similarly,
the Voronoi region of the second {estimate} in the transformed space is
the half-infinite open interval $\left(\hat{u}^m,\infty\right)$ along
the solution axis.

Having these preliminary results on hand, we now proceed with the
derivation, considering separately each of the altruistic approaches.
\subsection{Heterarchical Altruistic Estimation}
Setting the gradient of $\JZ$ [expressed in \eqref{eq: grad
    \JZ}] to zero yields
\begin{equation}\label{eq:Heter:eq_from_first_derivatives_prelim}
  \hat{\bm{\theta}}^{(i)} = \EZvorI\bm{\theta}, \quad i=1,2
\end{equation}
which shows that an optimal heterarchical {estimate} is locally
MMSE-optimal inside its Voronoi region. As such, it inherits the
properties of MMSE estimators inside that region. We have thus shown
that the two heterarchical estimators yield a CVT of $\thetaSpace$
(this has also been shown, for other Voronoi problems of similar
nature, in, e.g., \cite{CVT}).

Equations \eqref{eq:Heter:eq_from_first_derivatives_prelim} are
coupled and, generally, hard to solve in closed form. Efficient
algorithms for their \emph{numerical} solution {(even in the general
  case involving more than two facilities, e.g., Lloyd's algorithm)}
can be found in \cite{Voronoi_book,Voronoi_old,CVT}.
\begin{remark}
  Using the law of total probability and
  (\ref{eq:Heter:eq_from_first_derivatives_prelim}) yields
\begin{align}\label{eq:remark:MMSE generalization}
\expect (\bm{\theta}\givenMeas)  =
\PZvor(\mathcal{V}_1) \EZvorOne\bm{\theta}
+
\PZvor(\mathcal{V}_2) \EZvorTwo\bm{\theta}   =
\PZvor(\mathcal{V}_1) \hat{\bm{\theta}}^{(1)}
+
\PZvor(\mathcal{V}_2) \hat{\bm{\theta}}^{(2)}
\end{align}
which generalizes the fundamental theorem of MMSE estimation.  This
should come as no surprise, as the cost function
\eqref{eq:CostFunctionDefinition} generalizes the MSE cost function.
In fact, as shown in {Appendix~}\ref{sec:app_proof_opt_heter}, when the
norm-difference between the {estimates} tends to infinity, the
Voronoi region of the {estimate} possessing the larger norm tends
to a set of measure zero, whereas the other Voronoi region tends to a
set of measure one. In that case, (\ref{eq:remark:MMSE
  generalization}) yields that the (single) MMSE estimator is the
familiar global conditional mean.
\end{remark}

Sometimes it might be advantageous to calculate first $\hat u^m$, the
midpoint between {estimates} along the solution axis. To do that we use
(\ref{eq:def_u}) in (\ref{eq:Heter:eq_from_first_derivatives_prelim})
to obtain the following alternative form of
\eqref{eq:Heter:eq_from_first_derivatives_prelim}
\begin{equation}\label{eq:Heter:eq_from_first_derivatives}
\hat{\bm{\theta}}^{(i)} =
\transf\transp \EZvorI\bm{u}, 
\quad i=1,2.
\end{equation}
Now using \eqref{eq:uhatdef} in \eqref{eq:umdef} and substituting
\eqref{eq:Heter:eq_from_first_derivatives} results in the following
scalar equation
\begin{equation}\label{eq:Heter:altruism equation prelim}
  \hat{u}^m 
  = \frac{1}{2} (\EZvorOne u_1 + \EZvorTwo u_1)
\end{equation} 
which only depends on the marginal, conditional distribution of $u_1$.

Hereafter referred to as the {heterarchical altruism equation},
equation (\ref{eq:Heter:altruism equation prelim}) follows naturally
from the symmetric definition of the heterarchical estimation problem.
Using (\ref{eq: D1 is 1D in u}), \eqref{eq:Heter:altruism equation
  prelim} can be rewritten as
\begin{equation}\label{eq:Heter:altruism equation}
  \hat{u}^m
  = 
  \frac{1}{2} [\expect(u_1 \mid u_1<\hat{u}^m, \givenMeas) + \expect(u_1 \mid u_1>\hat{u}^m, \givenMeas)]
\end{equation}
revealing its dependence on the {truncated distribution} of $u_1$
given $\bm{Z}$ \cite[Chapter 22]{Greene:2012}.  We will use this
equation in solving the Gaussian case (Section \ref{sec:est_gauss}).

Finally, although the cost function can have an unbounded domain and
is not everywhere differentiable, we can still say something about the
existence of globally-optimal solutions. We do this in the next
theorem, proven in Appendix~{\ref{sec:app_proof_opt_heter}}. {In passing,
  we note that a similar theorem, related to the general case albeit
  assuming a compact parameter domain, is proven
  in~\cite[pp.~651-652]{CVT}.}
\begin{theorem}\label{theo: Heter optimality}
  There exists at least one globally-optimal heterarchical
  solution. All such solutions satisfy
  (\ref{eq:Heter:eq_from_first_derivatives_prelim}), and their
  (identical) cost is smaller than the MMSE.
\end{theorem}
\subsection{Hierarchical Altruistic Estimation}
Recall that, in this case, the first hierarchical estimator is the
MMSEE, so that only the second estimator needs to be found. Setting
the value of the gradient of $\JZ$ with respect to the second
{estimate} [expressed in \eqref{eq: grad \JZ}] to zero yields
\begin{equation}\label{eq:Hier:eq_from_first_derivatives_prelim}
  \hat{\bm{\theta}}^{(2)} = \EZvorTwo\bm{\theta} 
\end{equation}
rendering the second hierarchical {estimate} locally MMSE-optimal with
respect to its Voronoi region.

To calculate $\hat{u}^m$ for the hierarchical {estimates} we use
(\ref{eq:def_u}) in (\ref{eq:Hier:eq_from_first_derivatives_prelim})
to yield
\begin{equation}\label{eq:Hier:eq_from_first_derivatives}
\hat{\bm{\theta}}^{(2)} =
\transf\transp \EZvorTwo\bm{u}.
\end{equation}
Using \eqref{eq:uhatdef} in \eqref{eq:umdef} and substituting
\eqref{eq:Hier:eq_from_first_derivatives} yields
\begin{equation}\label{eq:Hier:altruism equation prelim}
\hat{u}^m
= \frac{1}{2} [\transf\expect( \bm{\theta}\givenMeas)
+\transf\hat{\bm{\theta}}^{(2)}]\transp \bm{e_{u_1}} 
= 
\frac{1}{2} [\expect( u_1\givenMeas)
+
\EZvorTwo u_1
].
\end{equation}

Equation (\ref{eq:Hier:altruism equation prelim}), which is a scalar
equation, is referred to as the {hierarchical altruism
  equation}. Notice that the first expectation on the RHS of
(\ref{eq:Hier:altruism equation prelim}) is with respect to the entire
sample space $\thetaSpace$ {(conditioned on $Z$)}, which follows
naturally from the fact that the first hierarchical estimator is the
MMSEE.

Using (\ref{eq: D1 is 1D in u}), (\ref{eq:Hier:altruism equation
  prelim}) can be rewritten as
\begin{equation}\label{eq:Hier:altruism equation}
  \hat{u}^m
  =
  \frac{1}{2} [ \expect(u_1 \givenMeas) 
+ \expect(u_1 \mid  u_1>\hat{u}^m, \givenMeas) ],
\end{equation}
revealing its dependence on the truncated conditional distribution of
$u_1$\cite[Chapter 22]{Greene:2012}.

We conclude with the next theorem, the proof of which is deferred to
Appendix~{\ref{sec:app_proof_opt_hier}}.
\begin{theorem}\label{theo: Hier optimality}
  There exists at least one globally-optimal hierarchical
  solution. All such solutions satisfy
  (\ref{eq:Hier:eq_from_first_derivatives_prelim}), and their
  (identical) cost is smaller than the MMSE.
\end{theorem}

%%% Local Variables:
%%% mode: latex
%%% TeX-master: "main"
%%% End:

\section{The Gaussian Case}\label{sec:est_gauss}
In this section we assume that the parameter and the measurements are
jointly Gaussian distributed, so that:
\begin{equation}
\bm{\theta}\mid \bm{Z} \sim \mathcal{N}(\bm{\theta}\mid \bm{Z}{,}\bm{\mu}{(\bm{Z})},\bm{R}{(\bm{Z})})
\end{equation}
where $\bm{\mu}$, the conditional mean, is the MMSEE. The conditional
covariance matrix $\bm{R}$ is assumed to be positive definite.  Let
the eigenvalues of $\bm{R}$ be
$\lambda_1 \geq \lambda_2 \geq \dots \geq \lambda_n$, and let their
corresponding unit-{norm} eigenvectors be
$\bm{v}_{\lambda_1},\bm{v}_{\lambda_2},\dots,\bm{v}_{\lambda_n}$.  We further denote
the conditional mean of $\bm{u}$ [as defined in (\ref{eq:def_u})]
{and its components} as
\begin{equation}\label{eq: mu_u def}
  \bm{\mu_u} \dfn \expect (\bm{u}\mid \bm{Z})
  { \equiv (\mu_{u_1} , \mu_{u_2}, \dots \mu_{u_n} )}.
\end{equation}

We next derive the {optimal} solution for each of the altruistic
estimation approaches.
\subsection{Heterarchical Altruistic Estimation}
\begin{theorem}
\label{Theo: Gauss Heter}
In the Gaussian case, the optimal altruistic heterarchical estimates
are
\begin{subequations}\label{eq: Heter estimators}
\begin{align}
  \hat{\bm{\theta}}^{(1)}_\Heter &= \bm{\mu} + \sqrt{\frac{2 \lambda_{1}}{\pi}} \bm{v}_{\lambda_{1}}
  \label{eq: Heter estimator 1}
  \\
  \hat{\bm{\theta}}^{(2)}_\Heter &= \bm{\mu} - \sqrt{\frac{2 \lambda_{1}}{\pi}} \bm{v}_{\lambda_{1}}
  \label{eq: Heter estimator 2}
\end{align}
\end{subequations} {and their estimation error covariances are
  identical to that of the MMSEE}.
\end{theorem}
\begin{proof}
  We begin by stating the following proposition, for the proof of
  which the reader is referred to
  {Appendix~\ref{sec:altruistic_eq}.}
\begin{prop}\label{prop:Heter altruistic}
  In the Gaussian case, the unique solution to the heterarchical
  altruism equation (\ref{eq:Heter:altruism equation prelim}) is
\begin{equation}\label{eq:Heter:um_is_mu1}
\hat{u}^m = \mu_{u_1}.
\end{equation}
\end{prop}
Using (\ref{eq:Heter:um_is_mu1}) in (\ref{eq: D1 is
  1D in u}) and noting the symmetry of the Gaussian distribution about
its mean yields
\begin{equation}\label{eq:Heter: Probability of V1 is half}
\PZvor(\mathcal{V}_1) = \Pr(u_1<\mu_{u_1} \givenMeas) = \PZvor(\mathcal{V}_2) = \frac{1}{2}
\end{equation}
which is a manifestation of the heterarchy in our problem.  Using
(\ref{eq:Heter:eq_from_first_derivatives}) and the law of total
probability yields 
\begin{align} \label{eq:theta_m is zero}
  \frac{\hat{\bm{\theta}}^{(1)}+\hat{\bm{\theta}}^{(2)}}{2} =
  \transf\transp ( \EZvorOne \bm{u} \PZvor(\mathcal{V}_1) 
  +
  \EZvorTwo \bm{u} \PZvor(\mathcal{V}_2)) =
  \transf\transp  \expect(\bm{u}\mid \bm{Z}) = \bm{\mu},
\end{align}
identifying the mid-point between the two {estimates} as the MMSE
estimate.  It follows that
\begin{equation}\label{eq:Heter: Delta is 2 times estimator}
  \Delta \bm{\hat{\theta}} = 2(\bm{\hat{\theta}}^{(2)} - \bm{\mu})
  = -2(\bm{\hat{\theta}}^{(1)} - \bm{\mu})
\end{equation}
which facilitates parameterizing the problem in terms of
$\Delta \bm{\hat{\theta}}$, {thus reducing the problem's degrees of
  freedom by half}.  The cost function can, therefore, be recast as
\begin{equation}\label{eq:Heter: \JZ as func of Deltas}
\JZ = \expect (
\norm{\bm{\mu}-\frac{1}{2}\Delta \bm{\hat{\theta}}-\bm{\theta}}^2 \wedge
\norm{\bm{\mu}+\frac{1}{2}\Delta \bm{\hat{\theta}}-\bm{\theta}}^2
\givenMeas ).
\end{equation}
Manipulating (\ref{eq:Heter: \JZ as func of Deltas}) results in
\begin{align}\label{eq:Heter:\JZ with traceR}
  \nonumber
  \JZ & = \expect [
      \norm{\bm{\mu}}^2 + \frac{1}{4}\norm{\Delta\bm{\hat{\theta}}}^2 
      + \norm{\bm{\theta}}^2 -2\langle\bm{\mu},\bm{\theta}\rangle
      +   (
      \Delta\bm{\hat{\theta}}\transp (\bm{\theta}-\bm{\mu})
      \wedge
      \Delta\bm{\hat{\theta}}\transp (\bm{\mu}-\bm{\theta})
      )
      \givenMeas ]
      \notag   \\ 
    & = \trace\bm{R} +  \frac{1}{4}\norm{\Delta\bm{\hat{\theta}}}^2
      - \expect(\abs{\Delta\bm{\hat{\theta}}\transp(\bm{\theta}-\bm{\mu})}
      \givenMeas )
\end{align}
where we have used the fact that $\min(-a,a) = -\abs{a}$ for any
$a\in\mathbb{R}$. Substituting
\begin{align}\label{eq:ttranspt}
   \Delta\bm{\hat{\theta}}\transp(\bm{\theta}-\bm{\mu}) 
    = \Delta\bm{\hat{\theta}}\transp \transf^T \transf
    (\bm{\theta}-\bm{\mu}) = (\transf \Delta\hat\theta)^T(u-\mu_u)
    = \norm{\Delta\hat\theta}(u_1 - \mu_{u_1})
\end{align}
in \eqref{eq:Heter:\JZ with traceR} yields
\begin{equation}\label{eq:finalJheter}
\JZ  =  \trace\bm{R} +  \frac{1}{4}\norm{\Delta\bm{\hat{\theta}}}^2
      - \norm{\Delta\bm{\hat{\theta}}} \expect(\abs{ u_1-\mu_{u_1}} 
      \givenMeas ).
\end{equation}
Explicitly expressing the central absolute first moment of the
Gaussian variable $u_1$ in \eqref{eq:finalJheter}\cite{folded} yields
\begin{align}\label{eq:Heter: \JZ final}
\JZ &= \trace\bm{R} + 
\frac{1}{4}\norm{\Delta\bm{\hat{\theta}}}^2
- \norm{\Delta\bm{\hat{\theta}}}\sqrt{\frac{2R_{u_1}}{\pi}}
\end{align}
with $R_{u_1}$ being the conditional variance of $u_1\givenMeas$.
The term $R_{u_1}$ depends only on the direction of $\Delta\bm{\hat\theta}$
(and not on its norm), because the rotation matrix \transf that maps
$\bm{\theta}$ into $\bm{u}$ is a function of that direction only. This
observation, then, means that  (\ref{eq:Heter: \JZ final}) is a
parametrization of $\JZ$ in terms of the norm and argument of
$\Delta\bm{\hat\theta}$. We thus proceed with finding the optimal norm
first.  Differentiating $\JZ$ with respect to
$\norm{\Delta\bm{\hat{\theta}}}$ and setting the derivative to zero
yields:
\begin{align}\label{eq:Heter: NormDelta}
  \norm{\Delta\bm{\hat{\theta}}} &= 2\sqrt{\frac{2R_{u_1}}{\pi}}.
\end{align}
Substituting (\ref{eq:Heter: NormDelta}) into (\ref{eq:Heter: \JZ
  final}) yields
\begin{equation}\label{Heter: \JZ after subbing norm}
  \JZ\big|_{\norm{\Delta\bm{\hat{\theta}}} = 2\sqrt{\frac{2R_{u_1}}{\pi}}} = \trace\bm{R} - \frac{2R_{u_1}}{\pi}.
\end{equation}
Therefore, minimizing $\JZ$ is equivalent to solving
\begin{equation}
  \max_{\Delta\hat{\bm{\theta}}} R_{u_1} \quad \text{such that}\quad  \norm{\Delta\bm{\hat{\theta}}} = 2\sqrt{\frac{2R_{u_1}}{\pi}}.
\end{equation}
To do that we write
\begin{equation}\label{eq: Ru1}
  R_{u_1}
   =  \bm{e_{\theta_1}}\transp (\transf \bm{R} \transf\transp) \bm{e_{\theta_1}}
  = \frac{\Delta\hat{\bm{\theta}}\transp }{\norm{\Delta\bm{\hat{\theta}}}}
    \bm{R}
    \frac{\Delta\hat{\bm{\theta}}}{\norm{\Delta\bm{\hat{\theta}}}}
\end{equation}
where $\bm{e}_{\theta_1}$ is the unit vector along the first standard
basis vector, so that the maximization problem becomes
\begin{equation}\label{eq:heter:max_prob_before_RayleighRitz}
\max_{ \Delta\hat{\bm{\theta}}}  
\frac{\Delta\hat{\bm{\theta}}\transp \bm{R}\Delta\hat{\bm{\theta}}}
    {\Delta\hat{\bm{\theta}}\transp \Delta\hat{\bm{\theta}}}.
\end{equation}

According to the Rayleigh-Ritz theorem \cite{Rayleigh-Ritz}, the
maximum in (\ref{eq:heter:max_prob_before_RayleighRitz}) is
$\lambda_1$, the largest eigenvalue of $\bm{R}$, and it is reached for
$\Delta\hat{\bm{\theta}}$ that is collinear with the eigenvector
$\bm{v}_{\lambda_{1}}$ of $\bm{R}$ corresponding to
$\lambda_{1}$. Thus, using (\ref{eq:Heter: NormDelta}), we have
\begin{equation}\label{eq:Heter:Delta theta optimal}
  \Delta\hat{\bm{\theta}} = 2\sqrt{\frac{2 \lambda_{1}}{\pi}} \bm{v}_{\lambda_{1}}
\end{equation}
which, with (\ref{eq:Heter: Delta is 2 times estimator}), then yields
\eqref{eq: Heter estimators}.  Moreover, using (\ref{Heter: \JZ after
  subbing norm}), the cost {obtained by using the candidate
  heterarchical estimators is}
\begin{equation}\label{eq: \JZ Heter cost}
J_\Heter = \trace \bm{R} - \frac{2}{\pi}\lambda_{1}
\end{equation}
which is identical among all candidate solutions and independent of
$\bm{v}_{\lambda_{1}}$.  Combining this fact with Theorem \ref{theo:
  Heter optimality}, which states that the candidate solutions include
at least one global solution, we conclude that all candidate solutions
are global minimizers.

{Finally, because $\hat{\bm{\theta}}^{(1)}_\Heter$,
  $\hat{\bm{\theta}}^{(2)}_\Heter$ differ from the MMSEE by a
  deterministic constant (given the measurements), their estimation
  error covariances are identical to that of the MMSEE.}
\end{proof}
In passing, we observe that, as $J_\Heter$ depends only on $\bm{R}$,
then, for the optimal estimators, $\J = \JZ = J_\Heter$.  {Also, as
  the solution requires only the computation of the largest eigenvalue
  and its corresponding eigenvector, efficient numerical algorithms,
  such as the power method, can be used in real-time applications
  involving high dimensionality.}
\subsection{Hierarchical Altruistic Estimation}
\begin{theorem}
\label{Theo: Gauss Hier}
In the Gaussian case, letting the first altruistic hierarchical
{estimate} be
\begin{equation}\label{eq: Hier first_estimator}
\bm{\hat{\theta}}^{(1)}_\Hier = \thetaMS = \bm{\mu}
\end{equation}
the optimal second {estimate} is
\begin{equation}
  \label{eq: Hier estimators}
  \bm{\hat{\theta}}^{(2)}_\Hier = \bm{\mu} +  \wHier \sqrt{\lambda_{1}}
  \bm{v}_{\lambda_{1}}  
\end{equation}
where $\wHier$ is defined
such that $\chi = \frac{1}{2}\wHier$ is the unique solution to
\begin{equation}\label{eq:Hier:chi_eq Theorem}
  \frac{\phi(\chi)}{2[1-\Phi(\chi)]}
  - \chi = 0.  
\end{equation}
In \eqref{eq:Hier:chi_eq Theorem}, $\phi$ and $\Phi$ are the standard
Gaussian probability density and cumulative distribution functions,
respectively. The estimation error covariance of
$\bm{\hat{\theta}}^{(2)}_\Hier$ is identical to that of the MMSEE.
\end{theorem}
\begin{remark}
  The approximate value of $\wHier$ is $1.224$ (see
  Appendix~\ref{sec:altruistic_eq}).  Notice that $\wHier$, which can
  be referred to as the standardized hierarchical shift of
  $\bm{\hat{\theta}}^{(2)}_\Hier$ from the conditional distribution
  mode (as it is the shift for the standardized case where
  $\lambda_{1}=1$), is bigger than its analogous heterarchical
  standardized shift, $\sqrt{\frac{2}{\pi}}$, as per
  Theorem~\ref{Theo: Gauss Heter}.  This shift difference can be
  explained by observing that in the hierarchical approach the first
  estimate is the mode, rendering a bigger shift from it
  probabilistically beneficial, compared with the heterarchical
  approach, where both estimates are already shifted from the mode in
  opposite directions.
\end{remark}
\begin{remark} In \eqref{eq: Hier estimators}, the sign of the second
  term on the RHS is arbitrary, because the sign of the eigenvector
  $\bm{v}_{\lambda_{1}}$ is arbitrary.
\end{remark}
\begin{proof}
  Using (\ref{eq: Hier first_estimator}) in
  (\ref{eq:Lemma_Deltas_definition}) yields
  \begin{equation}\label{eq:Hier: Delta is equal to second estimator
      minus MMSE}
\bm{\hat{\theta}}^{(2)} = \bm{\mu} + \Delta \bm{\hat{\theta}}
\end{equation}
so that the cost function can be written as
\begin{equation}\label{eq:hier_cost_gauss}
  \JZ = \expect (
    \norm{\bm{\mu} - \bm{\theta}}^2 \wedge
    \norm{\bm{\mu} + \Delta\hat{\bm{\theta}}-\bm{\theta}}^2
    \givenMeas ).
\end{equation}
Manipulating \eqref{eq:hier_cost_gauss} yields 
\begin{align}\label{eq:Hier:\JZ with traceR_prelim}
  & \JZ  =  \trace\bm{R} + \frac{1}{2}\norm{\Delta\hat{\bm{\theta}}}^2 
    - \Delta\hat{\bm{\theta}}\transp\expect (\bm{\theta}-\bm{\mu}\givenMeas)
    \notag  \\
  & \quad  + \expect \{ [
    -\frac{1}{2}\norm{\Delta\hat{\bm{\theta}}}^2 + \Delta\hat{\bm{\theta}}\transp(\bm{\theta}-\bm{\mu})]
    \wedge
    [\frac{1}{2}\norm{\Delta\hat{\bm{\theta}}}^2 -\Delta\hat{\bm{\theta}}\transp(\bm{\theta}-\bm{\mu})]
    \givenMeas \}  \notag \\
  &  = \trace\bm{R} + \frac{1}{2}\norm{\Delta\hat{\bm{\theta}}}^2
    - \expect (
    \abs{\Delta\hat{\bm{\theta}}\transp(\bm{\theta}-\bm{\mu}) -\frac{1}{2}\norm{\Delta\hat{\bm{\theta}}}^2 }
    \givenMeas )
\end{align}
where we have used $\expect[(\bm{\theta}-\bm{\mu})\mid\bm{Z}] = 0$.
Using \eqref{eq:ttranspt} in \eqref{eq:Hier:\JZ with traceR_prelim}
yields
\begin{align}\label{eq:Hier:\JZ with traceR}
  \nonumber
  & \JZ  = \trace\bm{R} + \frac{1}{2}\norm{\Delta\hat{\bm{\theta}}}^2
    - \norm{\Delta\hat{\bm{\theta}}} \expect (
    \abs{(u_1-\mu_{u_1}) -\frac{1}{2}\norm{\Delta\hat{\bm{\theta}}} }
    \givenMeas )
    \nonumber \\
  & \quad = \trace\bm{R} + \frac{1}{2}\norm{\Delta\hat{\bm{\theta}}}^2
  - \norm{\Delta\hat{\bm{\theta}}}
    [ -\frac{1}{2}\norm{\Delta\hat{\bm{\theta}}}(1-2\Phi(\frac{\norm{\Delta\hat{\bm{\theta}}}}{2\sqrt{R_{u_1}}})) 
    + 2\sqrt{R_{u_1}} \phi(\frac{\norm{\Delta\hat{\bm{\theta}}}}{2\sqrt{R_{u_1}}})
    ].
\end{align}
Equation (\ref{eq:Hier:\JZ with traceR}) is obtained by introducing
$\varphi\dfn u_1-(\mu_{u_1} +
\frac{1}{2}\norm{\Delta\hat{\bm{\theta}}})$, and calculating the
conditional mean of the folded variable
$\expect (\abs{\varphi} \givenMeas )$\cite{folded}.

To proceed with the minimization of $\JZ$, we need to calculate
$\norm{\Delta\hat{\bm{\theta}}}$.  To do that, we need to first solve
the hierarchical altruism equation (\ref{eq:Hier:altruism
  equation}). It turns out that this equation has no analytical
solution, even in the Gaussian case. Using $\wHier$, which was defined
implicitly in the Theorem, yields its unique solution (see
Appendix~\ref{sec:altruistic_eq}) as
\begin{equation}
                  \hat{u}^m = \mu_{u_1} + \frac{1}{2}\wHier
                  \sqrt{R_{u_1}}\label{eq:Hier: um result}.
\end{equation}
To calculate $\norm{\Delta\bm{\hat\theta}}$ we use the rotation
transformation \transf:
\begin{equation}
  \norm{\Delta\bm{\hat\theta}} = \norm{\transf \Delta\bm{\hat\theta}} = \norm{\Delta\bm{\hat u}}.
\end{equation}
In the transformed parameter space, both transformed {estimates}
reside along the solution axis $\bm{e_{u_1}}$, such that $\hat{u}^m$ is the
midpoint between them along that axis.  Hence
\begin{equation}
  \norm{\Delta\bm{\hat u}} = \Delta\hat u_1 = 2(\hat u^m - \mu_{u_1})
\end{equation}
where $\Delta\hat u_1$ is the component of $\Delta \bm{\hat u}$ along
$\bm{e_{u_1}}$.  Using \eqref{eq:Hier: um result} then yields
\begin{equation}\label{eq:deltatheta}
  \norm{\Delta\bm{\hat\theta}} = \wHier \sqrt{R_{u_1}}.
\end{equation}
Using \eqref{eq:deltatheta} in (\ref{eq:Hier:\JZ with traceR}) yields
\begin{align}
  \nonumber
  \JZ\big|_{\norm{\Delta\bm{\hat{\theta}}} =
  \wHier \sqrt{R_{u_1}}}
  & = \trace\bm{R} -
    [ 2\phi(\frac{\wHier}{2})
    - \wHier (1 - \Phi(\frac{\wHier}{2}))
    ] \wHier R_{u_1}
  \\
  &= \trace\bm{R} - \phi(\frac{\wHier}{2}) \wHier R_{u_1}
  	\label{eq:Hier: \JZ final}
\end{align}
where the last equality results from using \eqref{eq:Hier:chi_eq} with
\eqref{eq:138a} (see Appendix~\ref{sec:altruistic_eq}). Moreover, as
analytically proved in Appendix~\ref{sec:altruistic_eq},
$\frac{1}{2}\wHier \in (0,\sqrt{3})$, hence
$\phi(\frac{\wHier}{2}) \wHier > 0$. Thus, \eqref{eq:Hier: \JZ final}
leads to a maximization problem identical to
\eqref{eq:heter:max_prob_before_RayleighRitz}, obtained in the
heterarchical problem. Adopting the solution of that problem (together
with the definition of $\wHier$) yields \eqref{eq: Hier estimators}.
Moreover, {the cost obtained from using the candidate hierarchical
  estimates is}
\begin{equation}\label{eq: \JZ Hier cost}
  J_\Hier = \trace\bm{R} - \phi(\frac{\wHier}{2}) \wHier\lambda_{1}
\end{equation}
and that cost is identical among all candidate solutions, and
independent of $\bm{v}_{\lambda_{1}}$.  Combining this observation
with Theorem \ref{theo: Hier optimality}, which states that the
candidate solutions include at least one global solution, renders all
candidate solutions global minimizers.

Finally, because $\hat{\bm{\theta}}^{(2)}_\Hier$ differs from the
MMSEE by a deterministic constant (given the measurements), its
estimation error covariance is identical to that of the MMSEE.
\end{proof}
Similarly to the heterarchical problem, here too only the largest
eigenvalue and its corresponding eigenvector need to be calculated.
\subsection{Cost Reduction}\label{sec:Results}
To assess the benefit of the altruistic cooperative methodology, we
compare its achievable MSE cost~(\ref{eq:CostFunctionDefinition}) to
the MMSE baseline cost achieved by using two identical MMSE
{estimates},
\begin{equation}
\JMS \dfn \expect 
\norm{\thetaMS-\bm{\theta}}^2.
\end{equation}
In the Gaussian case $\JMS = \trace {\bm{R}}$.

For both approaches we define the relative cost reduction as
\begin{equation}\label{eq: J_Heter reduction def}
  \Upsilon_\text{MTHD}
  \dfn
  1 - \frac{J_\text{MTHD}}{\JMS}, \quad \text{MTHD} = \Heter \text{ or }
  \Hier.
\end{equation}

In the heterarchical approach, (\ref{eq: \JZ Heter cost}) yields
\begin{align}\label{eq: Upsilon_Heter}
\Upsilon_\Heter
&=
\frac{\frac{2}{\pi}\lambda_1}
{\sum_{i=1}^n \lambda_i}.
\end{align}
Since $\lambda_1 \geq \lambda_2 \geq \dots \geq \lambda_n$, then, for
a given $\lambda_1$,
\begin{equation}\label{eq: Upsiln_Heter upper bound}
\sup_{\lambda_2,\dots,\lambda_n} \Upsilon_\Heter
=
\lim\limits_{\frac{\lambda_2}{{\lambda_1}} \to 0} \Upsilon_\Heter
=
\frac{2}{\pi}
\end{equation}
and
\begin{equation}\label{eq: Upsiln_Heter lower bound}
\min_{\lambda_2,\dots,\lambda_n} \Upsilon_\Heter
=
\Upsilon_\Heter \big|_{\lambda_2=\dots=\lambda_n = \lambda_1} 
=
\frac{2}{n\pi}
\end{equation}
which gives
\begin{equation}
\frac{2}{n\pi} \leq \Upsilon_\Heter < \frac{2}{\pi}.
\end{equation}

In the hierarchical case,
% substituting the approximate value of
% $\wHier$  \eqref{eq:56b} in (\ref{eq: \JZ Hier cost}) yields
\begin{align}\label{eq: Upsilon_Hier}
\Upsilon_\Hier
= 
\frac{\wHier \phi(\frac{\wHier}{2}) \lambda_1}
{
\sum_{i=1}^n \lambda_i
}
\end{align}
whence
\begin{equation}
  \frac{\wHier \phi(\frac{\wHier}{2})}{n} \leq \Upsilon_\Hier < \wHier \phi(\frac{\wHier}{2}).
\end{equation}

Notice that in both approaches the best achievable relative reduction
corresponds to $\frac{\lambda_1}{\lambda_2}\to\infty$, whereas the
worst achievable reduction corresponds to
$\lambda_1=\lambda_2=\ldots=\lambda_n$. This is so because in both
approaches the two {estimates} are dispersed along the eigenvector
that corresponds to $\lambda_{1}$. Thus, the benefit gained from
dispersing the {estimates} is biggest when the variance in that
direction is largest compared with the other variances. When the
variances in all directions are equal, the benefit assumes its
smallest possible value. It is also noted that the benefit shrinks
when the dimension of the system increases, because there are only two
{estimates}, distributed along one direction. Nevertheless, even in
high dimensional cases, if one direction dominates the others in terms
of its variance, still the reduction can be significant, which means
that the altruistic approaches become appealing in cases involving
ill-conditioned covariance matrices (characterized by large condition
numbers).

To demonstrate the effect of the problem's dimensionality on the cost
function reduction, the upper and lower cost reduction bounds for each
approach are depicted in Fig.~\ref{fig:J_Reduction_Bounds}. In a
scalar problem ($n=1$) the lower and upper bounds coincide, yielding a
unique value for the reduction. At higher dimensions the best
achievable gains are identical to those obtained for the scalar
problem, whereas the worst achievable reductions diminish with the
increasing dimension.
\begin{figure}[tbh]
\centering
\includegraphics[width=\linewidth]{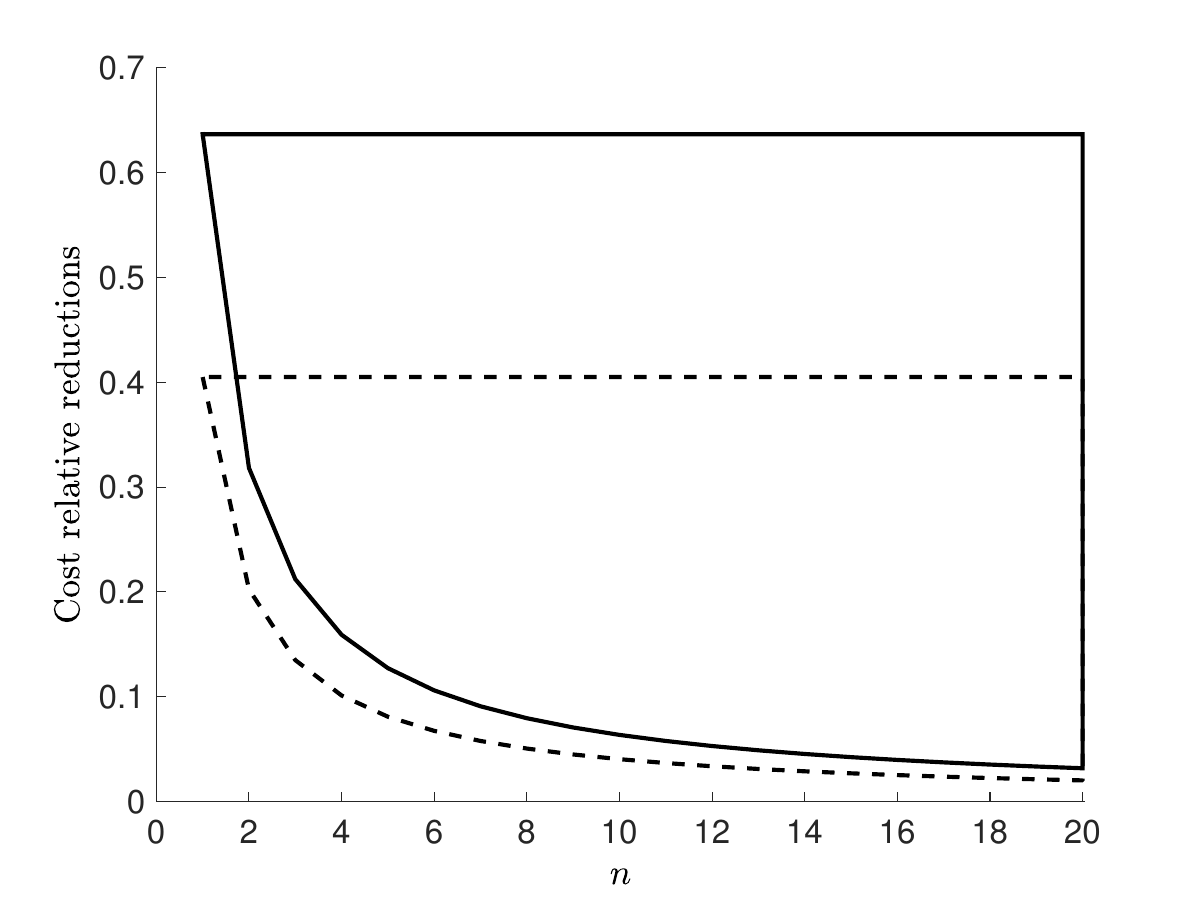}
\caption{Relative cost reduction bounds vs problem dimension. Solid
  lines: heterarchical ($\Upsilon_\Heter$).  Dashed lines:
  hierarchical ($\Upsilon_\Hier$).  }
\label{fig:J_Reduction_Bounds}
\end{figure}

\subsection{Numerical Illustration}\label{sec:illus}
We now numerically illustrate the behavior of the cost function $\JZ$,
in order to demonstrate the validity of the solutions in the Gaussian
case, and highlight some properties of altruistic estimation problems
and their solutions.  We focus on $\JZ$, rather than on $J$, because,
as we have previously shown, the choice of the optimal estimators is
independent of the measurements.

Let $\theta \givenMeas \sim \mathcal{N}(0,100)$.  The
heterarchical and hierarchical estimates are
\begin{equation}\label{eq:1dheter}
  \hat{\theta}^{(1)}_{\Heter} = 
  10\sqrt{\frac{2}{\pi}} \approx 7.979, \qquad \hat{\theta}^{(2)}_{\Heter} =
-\hat{\theta}^{(1)}_{\Heter} 
\end{equation}
and
\begin{equation}\label{eq:1dhier}
  \hat{\theta}^{(1)}_{\Hier} = \mu = 0, \qquad \hat{\theta}^{(2)}_{\Hier} = \pm 10 w_\Hier \approx \pm 12.240.
\end{equation}
The costs for the MMSE (two egoistic estimates), heterarchical, and
hierarchical approaches are: $\JMS = 100$, $J_{\Heter} \approx 36.338$,
and $J_{\Hier} \approx 59.5$, respectively.

Figure~\ref{fig:Example_Gaussian_1D} shows equilevel contour lines of
the cost function, computed at each node
$(\hat{\theta}^{(1)}, \hat{\theta}^{(2)})$ of a
$500\times 500$ grid of the two {estimates}. The cost is approximated
as the mean of $10^5$ samples drawn from the given parameter
{conditional} distribution. The contours are distributed
logarithmically, so that they are denser around lower values of
{$\JZ$}. The optimal heterarchical and hierarchical estimates,
\eqref{eq:1dheter} and \eqref{eq:1dhier}, respectively, are
superimposed on the contour plot as squares {and circles},
respectively.  The figure exhibits the tendency of {$\JZ$} to
infinity when both {estimates} tend to infinity in absolute values
(as per Lemma {1} in~{\ref{sec:app_proof_opt_heter}}). On the other
hand, when one of the {estimates} tends to infinity in absolute
value and the other {estimate} remains finite (which renders the
infinite {estimate} irrelevant in the computation of the cost),
the lowest value of {$\JZ$} is the MMSE {estimate}, which is
achieved when the finite {estimate} is the MMSEE.
\begin{figure}[tbh]
  \centering
  \includegraphics[width=\linewidth]{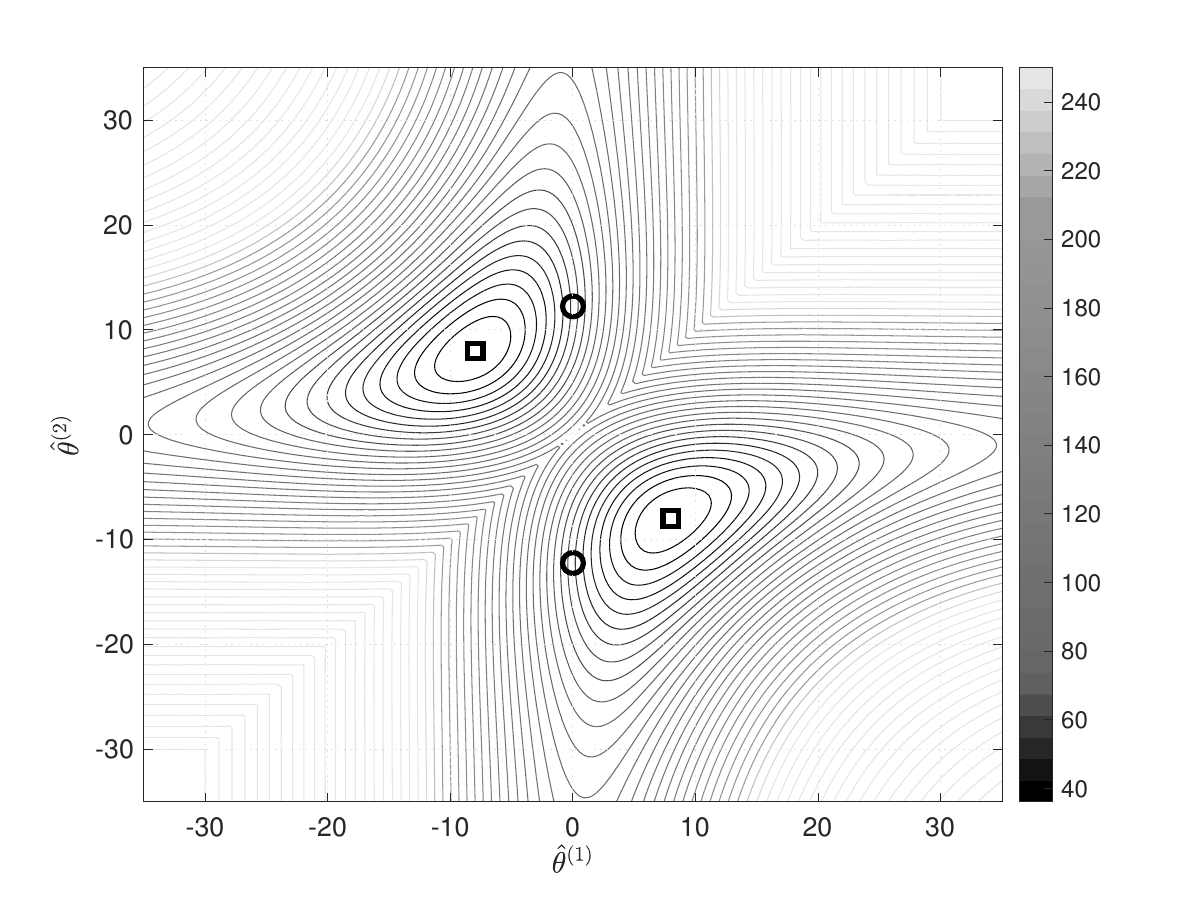}
  \caption{Equilevel cost contours of {$\JZ$} in the 1-D Gaussian
    example of subsection \ref{sec:illus}. Square markers:
    heterarchical solutions. Circle markers: hierarchical solutions.}
\label{fig:Example_Gaussian_1D}
\end{figure}

Figure~\ref{fig:Example_Gaussian_1D} exhibits two reflection
symmetries of the function {$\JZ$}, 1) about the mirror line
$\hat{\theta}^{(1)} = \hat{\theta}^{(2)}$, and 2) about the
normal to the line $\hat{\theta}^{(1)}=\hat{\theta}^{(2)}$
through the origin. The first symmetry expresses the symmetric nature
of {$\JZ$} with respect to its arguments (the two
{estimates}). The second expresses the symmetric nature of the
Gaussian distribution.  The figure also clearly demonstrates the
non-differentiability of {$\JZ$} on the reflection axis
$\hat{\theta}^{(1)}=\hat{\theta}^{(2)}$.  The figure
exhibits two (reflective) global minima of the function {$\JZ$},
precisely at the analytically calculated heterarchical {estimates}
~\eqref{eq:1dheter}, with cost agreeing with the analytically computed
cost. The two (reflective) optimal hierarchical {estimates}
{of the second hierarchical estimator} are positioned at the
 {minima} of {$\JZ$} along the
constraint line $\hat{\theta}^{(1)}=0$, which coincide with the
analytically calculated hierarchical
{estimates}~\eqref{eq:1dhier}, with cost agreeing with the
analytically computed cost. As could be expected, at both
{minima} the {$\JZ$} contour lines are
tangent to the constraint.

%%% Local Variables:
%%% mode: latex
%%% TeX-master: "main"
%%% End:

\section{Conclusions}\label{sec:conclusions}
We have proposed an estimation methodology for optimal cooperation
between two in\-formation-sharing agents, that is based on the notion
of altruism. The methodology is suited for scenarios {that can benefit
  from the existence of two opportunities to estimate, such as}
scenarios involving two cooperating {agents} that have one global
mission that is accomplished even if only one of the {agents} provides
a satisfactory estimate {(using its own, local estimator)}. In the
proposed approach, the two {agents} do not yield an identical optimal
estimate, but, rather, at least one of them sacrifices its own
estimation performance by providing a sub-optimal estimate. The
benefit of the proposed scheme is an improvement in the overall
estimation performance, measured by a global mean squared error
criterion.

Two approaches of altruistic cooperation are proposed. In the
heterarchical approach, both estimators are altruistic, which yields
two sub-optimal estimates that are different than the (egoistic) MMSE
estimate. In the hierarchical approach the first {agent} maximizes its
performance egoistically without considering the other estimator, thus
computing the MMSE estimate; the second {agent} maximizes the global
performance measure while taking into account the presence of the
first (MMSE-optimal) agent.  Implicit and coupled equations are
derived for the design of the estimators in both approaches.  In the
Gaussian case, explicit optimal solutions are provided, that require,
in both approaches, only the calculation of the largest eigenvalue of
the conditional covariance matrix of the parameter and its
corresponding eigenvector.  These results can also be viewed as
analytical solutions to two well-known Voronoi serviceability problems
in the two-facility case.

In the Gaussian case, it is shown that the improvement in the overall
performance (relative to naive MMSE estimation) depends on the
dimension of the problem and on the spread of the spectrum of the
conditional covariance matrix.  In general, the larger the dimension
of the problem, the smaller the improvement that can be expected using
the proposed cooperative estimation approach. On the other hand, the
proposed altruistic approaches are especially appealing in (even
high-dimensional) ill-conditioned estimation problems.

%%% Local Variables:
%%% mode: latex
%%% TeX-master: "main"
%%% End:

\begin{appendices}
\section{Proof of Theorem~\ref{theo: Heter optimality}}\label{sec:app_proof_opt_heter}
\begin{proof}
  Following the reasoning of the estimator derivation in
  Section~\ref{sec:est_general}, the proof follows from minimizing the
  measurement-conditioned cost, $\JZ$, defined in \eqref{eq:
    definition of \JZ}. Let
  $\hat{\bm{\theta}} \in (\augmentedSpaceRealization)$ and
  $\JZ (\hat{\bm{\theta}})$ denote the augmented estimate
  vector
  $[ (\hat{\bm{\theta}}^{(1)})^T, (\hat{\bm{\theta}}^{(2)})^T ]^T $
  and the value of the cost \eqref{eq: definition of \JZ}
  computed with the pair of estimates comprising
  $\hat{\bm{\theta}}$, respectively.
  
  For the reader's convenience, we first provide an overview of the
  proof's main stages:
  \begin{enumerate}
  \item We tessellate the augmented estimate vector space
    $(\augmentedSpaceRealization)$ into two parts: an internal part
    (where $\hat{\bm{\theta}}$ is bounded), and its complementary
    external part.
  \item We prove Lemma \ref{lem: Heter: \J>MMSE-eps}, showing that
    choosing the internal part to be large enough guarantees that for
    any point $\hat{\bm{\theta}}$ in the external part, the cost
    $\JZ (\hat{\bm{\theta}})$ is arbitrarily close to the MMSE.
  \item We prove Lemma \ref{lem: Heter: min of \J in the ball},
    showing that choosing the internal part to be sufficiently large
    guarantees that there is at least one minimum point (satisfying
    \eqref{eq:Heter:eq_from_first_derivatives_prelim}) in the internal
    part, for which the cost $\JZ (\hat{\bm{\theta}})$ is
    strictly smaller than the MMSE.
  \item We build a final internal part, which is sufficiently large to
    guarantee that any point in the external part yields a cost
    strictly higher than that of any minimum of the internal part.
\end{enumerate}

Moving on to the proof itself, define
\begin{equation}\label{eq: defi: S}
\Sset^a
\dfn
\{
\hat{\bm{\theta}} \mid %%\in \augmentedSpaceRealization:
 \norm{\hat{\bm{\theta}}^{(1)}} \vee     \norm{\hat{\bm{\theta}}^{(2)}}
\leq a
\}, \quad a\in\PositiveReals
\end{equation}
%%for some real number $a>0$.
%%
where $\alpha \vee \beta \dfn \max(\alpha, \beta)$ for some
$\alpha, \beta \in\mathbb{R}$. We begin by proving the following two
lemmas.
\begin{lemma}\label{lem: Heter: \J>MMSE-eps}
  For any $\epsilon\in\PositiveReals$ there exists
  $a(\epsilon)\in\PositiveReals$ such that if
  $\hat{\bm{\theta}} \not\in
  \Sset^{a(\epsilon)}$
  then $\JZ (\hat{\bm{\theta}})  > \JMS-\epsilon$.
\end{lemma}
\begin{proof}
  We begin with a brief overview of the proof of this lemma:
\begin{enumerate}
\item Choosing some (temporary) edge $a_0(\epsilon)$ to the internal
  part $\Sset^{a(\epsilon)}$, we show in Proposition~\ref{prop:a_not}
  that when the norms of both estimates are beyond $a_0(\epsilon)$,
  the lemma is satisfied.
\item Next, we show in Proposition~\ref{prop:D} that if the distance
  between the norms of both estimates is large enough, the lemma is
  satisfied.
\item Finally, we choose a (bigger) edge $a(\epsilon)$ for which, if
  $\hat{\bm{\theta}} \not\in \Sset^{a(\epsilon)}$, then either of the
  conditions of Proposition~\ref{prop:a_not} or
  Proposition~\ref{prop:D} must be satisfied.
\end{enumerate}
	
Using the law of total probability and the triangle inequality in
\eqref{eq: definition of \JZ} yields, for any $r\in\PositiveReals$,
\begin{align}\label{eq: \J for \J>MMSE second case}
\JZ (\hat{\bm{\theta}})   &= \expect ( 
\norm{\hat{\bm{\theta}}^{(1)}-\bm{\theta}}^2 \wedge
\norm{\hat{\bm{\theta}}^{(2)}-\bm{\theta}}^2
\givenMeas, {\norm{\bm{\theta}} \leq r})
\Pr(\norm{\bm{\theta}} \leq r \givenMeas)
\notag \\
&  +
\expect ( 
\norm{\hat{\bm{\theta}}^{(1)}-\bm{\theta}}^2 \wedge
\norm{\hat{\bm{\theta}}^{(2)}-\bm{\theta}}^2
\givenMeas, {\norm{\bm{\theta}} > r} )
\Pr(\norm{\bm{\theta}} > r \givenMeas)
\notag \\
& \geq
\expect ( 
\norm{\hat{\bm{\theta}}^{(1)}-\bm{\theta}}^2 \wedge
\norm{\hat{\bm{\theta}}^{(2)}-\bm{\theta}}^2
\givenMeas, {\norm{\bm{\theta}} \leq r} )
\Pr(\norm{\bm{\theta}} \leq r \givenMeas) \notag \\
&\geq \expect [ 
(\norm{\hat{\bm{\theta}}^{(1)}}-\norm{\bm{\theta}})^2 \wedge
(\norm{\hat{\bm{\theta}}^{(2)}}-\norm{\bm{\theta}})^2
\givenMeas,  {\norm{\bm{\theta}} \leq r} ]
\notag \\ 
& \quad \times
\Pr(\norm{\bm{\theta}} \leq r \givenMeas).
\end{align}
We now define
\begin{equation}\label{eq:r}
  r \dfn \argument\limits_{\rho\in\PositiveReals} \{\Pr(\norm{\bm{\theta}} \leq \rho \givenMeas) = \frac{1}{2}\}.
\end{equation}
Also, observing that $\epsilon < \JMS$ (as otherwise the lemma holds
trivially since $\JZ>0$), and using the assumption that the
second moment of the joint distribution of $\bm{\theta}$ and $\bm{Z}$
is finite, we set
\begin{equation}\label{eq:a in Lemma 1}
a_0(\epsilon) \dfn 1 + r + \sqrt{2(\JMS-\epsilon)}.
\end{equation}
We first state and prove the following propositions.
\begin{prop} \label{prop:a_not}
  If both $\norm{\hat{\bm{\theta}}^{(1)}} > a_0(\epsilon)$ and
  $\norm{\hat{\bm{\theta}}^{(2)}} > a_0(\epsilon)$, then
  $\JZ (\hat{\bm{\theta}})>\JMS-\epsilon$.
\end{prop}
\begin{proof}
  Employing definitions \eqref{eq:r} and \eqref{eq:a in Lemma 1} in
  (\ref{eq: \J for \J>MMSE second case}) yields
\begin{align}\label{eq: \J for \J>MMSE second case 2}
 \JZ (\hat{\bm{\theta}}) & \geq [ 
  (\norm{\hat{\bm{\theta}}^{(1)}}-r)^2 \wedge
  (\norm{\hat{\bm{\theta}}^{(2)}}-r)^2
  ]
  \Pr(\norm{\bm{\theta}} \leq r \givenMeas) \notag \\
 & >  
[a_0(\epsilon)-r]^2
\Pr(\norm{\bm{\theta}} \leq r \givenMeas)
\notag \\
&=  
\frac{1}{2}[1+\sqrt{2(\JMS-\epsilon)}]^2
>
\JMS-\epsilon.
\end{align}
\end{proof}
Prior to presenting the next proposition we now assume that
$\norm{\hat{\bm{\theta}}^{(2)}} \ge \norm{\hat{\bm{\theta}}^{(1)}}$ and define
% \begin{equation}\label{eq: defi: d}
$
\mathcal{D} \dfn
\norm{\hat{\bm{\theta}}^{(2)}}
-\norm{\hat{\bm{\theta}}^{(1)}}
$. 
% \end{equation}
\begin{prop}\label{prop:D}
  For any number $\epsilon\in\PositiveReals$ there exists a number
  $L(\epsilon)\in\PositiveReals$ such that if
  $\mathcal{D} > L(\epsilon)$ then $\JZ > \JMS - \epsilon$.
\end{prop}
\begin{proof}
  Define
  $\varphi(\hat{\bm{\theta}}^{(1)}) \dfn
  \expect[\norm{\hat{\bm{\theta}}^{(1)}-\bm{\theta}}^2 \givenMeas]$.
  Then, since $\VorEdge$ is a set of measure zero, the law of total
  probability yields
\begin{equation}\label{eq: MSE of estimator 1}
  \varphi(\hat{\bm{\theta}}^{(1)}) = \EZvorOne
  \norm{\hat{\bm{\theta}}^{(1)}-\bm{\theta}}^2 \PZvor(\mathcal{V}_1) +
  \EZvorTwo \norm{\hat{\bm{\theta}}^{(1)}-\bm{\theta}}^2
 \PZvor(\mathcal{V}_2).
\end{equation}
Equation (\ref{eq: MSE of estimator 1}) yields
\begin{equation}\label{eq: total prob. on first estim.} 
\EZvorOne  
\norm{\hat{\bm{\theta}}^{(1)}-\bm{\theta}}^2 
  \PZvor(\mathcal{V}_1)
=
\varphi(\hat{\bm{\theta}}^{(1)})
 -
\EZvorTwo  
\norm{\hat{\bm{\theta}}^{(1)}-\bm{\theta}}^2 
 \PZvor(\mathcal{V}_2).
\end{equation}
Expressing $\JZ$ using the law of total probability and using (\ref{eq:
  total prob. on first estim.}) yields
\begin{equation}\label{eq: \J with f1 f2}
  \JZ  =  \varphi(\hat{\bm{\theta}}^{(1)})  - f_1  + f_2
\end{equation}
where we have defined
\begin{gather}
\label{eq: defi: f1}
f_1
\dfn
\EZvorTwo  
\norm{\hat{\bm{\theta}}^{(1)}-\bm{\theta}}^2 
\PZvor(\mathcal{V}_2)
\\
\label{eq: defi: f2}
f_2 \dfn \EZvorTwo  
\norm{\hat{\bm{\theta}}^{(2)}-\bm{\theta}}^2 
\PZvor(\mathcal{V}_2).
\end{gather}
 Clearly, $f_1\geq0$ and $f_2\geq0$.  Using (\ref{eq:defi:D2}) and the monotonicity of the probability measure, we have
\begin{align}\label{eq: app2 eq1}
\nonumber
& \PZvor(\mathcal{V}_2)
= \Pr ({\norm{\hat{\bm{\theta}}^{(1)}-\bm{\theta}}^2
-
\norm{\hat{\bm{\theta}}^{(2)}-\bm{\theta}}}^2
> 0 \givenMeas)
\\ \nonumber
& \quad =
\Pr\{\norm{\hat{\bm{\theta}}^{(1)}}^2
- \norm{\hat{\bm{\theta}}^{(2)}}^2
+ 2 \bm{\theta}\transp(\hat{\bm{\theta}}^{(2)}-\hat{\bm{\theta}}^{(1)})
> 0  \givenMeas\}
\\ \nonumber
& \quad \leq
\Pr \{\norm{\hat{\bm{\theta}}^{(1)}}^2
- \norm{\hat{\bm{\theta}}^{(2)}}^2
+ 2 \abs{\bm{\theta}\transp(\hat{\bm{\theta}}^{(2)}-\hat{\bm{\theta}}^{(1)})}
> 0  \givenMeas \}
\\
& \quad \leq
\Pr \{\norm{\hat{\bm{\theta}}^{(1)}}^2
- \norm{\hat{\bm{\theta}}^{(2)}}^2
+ 2 \norm{\bm{\theta}}
\norm{\hat{\bm{\theta}}^{(2)}-\hat{\bm{\theta}}^{(1)}}
> 0 \givenMeas \}
\end{align}
where the last inequality follows from the Cauchy-Schwarz inequality.
Using the triangle inequality yields
\begin{align}
\nonumber
& \PZvor(\mathcal{V}_2)
\leq
\Pr \{\norm{\hat{\bm{\theta}}^{(1)}}^2
- \norm{\hat{\bm{\theta}}^{(2)}}^2
\\ 
\nonumber
& \qquad\qquad + 2 \norm{\bm{\theta}}
(\norm{\hat{\bm{\theta}}^{(1)}}
+ \norm{\hat{\bm{\theta}}^{(2)}})
> 0  \givenMeas \}
\\ 
& \qquad =
\Pr \{
\norm{\hat{\bm{\theta}}^{(1)}}
- \norm{\hat{\bm{\theta}}^{(2)}}
+ 2 \norm{\bm{\theta}}
> 0 \givenMeas \}
=
 \Pr \{\norm{\bm{\theta}}
 > \frac{1}{2}  \mathcal{D}  \givenMeas
  \}
\end{align}
showing that when $\mathcal{D} \to \infty$, $\mathcal{V}_2$ becomes a
set of measure zero. Thus $\PZvor(\mathcal{V}_1) \to 1$ and, hence, the
LHS of (\ref{eq: total prob. on first estim.}) satisfies
\begin{equation}\label{eq: E1 becomes MSE}
\lim\limits_{\mathcal{D} \to \infty}
\EZvorOne ( 
\norm{\hat{\bm{\theta}}^{(1)}-\bm{\theta}}^2 
) \PZvor(\mathcal{V}_1)
=
 \varphi(\hat{\bm{\theta}}^{(1)})
\end{equation}
yielding
$ \lim_{\mathcal{D} \to \infty}
f_1
= 0$.
By the definition (\ref{eq:defi:D2}) of the Voronoi region $\mathcal{V}_2$,
$ f_2 \leq f_1 $,
implying
$ \lim_{\mathcal{D} \to \infty}
f_2
= 0$.
Defining 
$ f \dfn f_1 - f_2 $,
 (\ref{eq: \J with f1 f2}) becomes
\begin{equation}
\JZ = \varphi(\hat{\bm{\theta}}^{(1)}) - f 
\end{equation}
where $f \geq 0$ and 
$\lim_{\mathcal{D} \to \infty}
f
= 0$.
Since, due to the fundamental theorem of MMSE estimation,
$\JMS \leq \varphi(\hat{\bm{\theta}}^{(1)})$, we have
\begin{equation}
\JZ \geq \JMS - f 
\end{equation}
which yields the  proposition.
\end{proof}
Now set 
\begin{equation}\label{eq: a choice first lemma}
a(\epsilon) \dfn 1 + a_0(\epsilon) + L(\epsilon).
\end{equation}
If
$\hat{\bm{\theta}} \not\in \Sset^{a(\epsilon)}$
then either
\begin{gather}
\label{eq: Heter Lemma: second case}
  \norm{\hat{\bm{\theta}}^{(2)}} 
  > a(\epsilon)
  \text{ and }
  \norm{\hat{\bm{\theta}}^{(1)}} 
  \leq a(\epsilon) \\
\intertext{or}
\label{eq: Heter Lemma: first case}
\norm{\hat{\bm{\theta}}^{(1)}} > a(\epsilon)
\text{ and }
\norm{\hat{\bm{\theta}}^{(2)}} > a(\epsilon).
\end{gather}
Assume, first, that \eqref{eq: Heter Lemma: second case} holds.  If
$\norm{\hat{\bm{\theta}}^{(1)}} > a_0(\epsilon) $ then
$\norm{\hat{\bm{\theta}}^{(2)}} > a_0(\epsilon) $ and the lemma is
proved based on Proposition~\ref{prop:a_not}.
Conversely, if $\norm{\hat{\bm{\theta}}^{(1)}} \leq a_0(\epsilon) $,
then, since $\norm{\hat{\bm{\theta}}^{(2)}} > a(\epsilon)$,
\begin{equation}\label{eq: Theorem: ratio of norms of in r ball}
  \mathcal{D} = \norm{\hat{\bm{\theta}}^{(2)}} - \norm{\hat{\bm{\theta}}^{(1)}}
  \geq  a(\epsilon) - a_0(\epsilon)
  =  1 + L(\epsilon)   > L(\epsilon)
\end{equation}
and the lemma follows from Proposition~\ref{prop:D}.

If \eqref{eq: Heter Lemma: first case} holds, then \eqref{eq: a choice
  first lemma} implies that both
$\norm{\hat{\bm{\theta}}^{(1)}} > a_0(\epsilon)$ and
$\norm{\hat{\bm{\theta}}^{(2)}} > a_0(\epsilon)$, and the lemma
follows from Proposition~\ref{prop:a_not}.
\end{proof}
\begin{lemma}\label{lem: Heter: min of \J in the ball}
  There exists a number $b_0\in\PositiveReals$ such that for any
  $b \geq b_0$, $\JZ$ attains at least one minimum in $\Sset^b$, at a
  point where $\JZ < \JMS$. All such points satisfy
  (\ref{eq:Heter:eq_from_first_derivatives_prelim}).
\end{lemma}
\begin{proof}
  Assuming that the first moment of the joint distribution of
  $\bm{\theta}$ and $\bm{Z}$ is finite, let
  $\hat{\bm{\zeta}}^{(1)} \dfn \thetaMS$ and
  $\hat{\bm{\zeta}}^{(2)} \dfn c \bm{e_{\theta_1}}$ {({recalling that}
    $\bm{e_{\theta_1}} \in \mathbb{R}^n$ is the unit vector along the
    first standard basis vector of $\thetaSpace$}), with $c\in (0, 1)$
  such that $\hat{\bm{\zeta}}^{(2)} \neq \thetaMS$. To show that
  $\JZ(\hat{\bm{\zeta}})<\JMS$, where $\hat\zeta$ comprises the
  estimators $\hat{\bm{\zeta}}^{(1)}$ and $\hat{\bm{\zeta}}^{(2)}$, we
  use the law of total probability to rewrite \eqref{eq: definition of
    \JZ} as
\begin{equation}\label{eq: app: AAAxi}
  \JZ(\hat{\bm{\zeta}})  = \EZvorOneZeta ( 
  \norm{\hat{\bm{\zeta}}^{(1)}-\bm{\theta}}^2 
  ) \PZvor(\mathcal{V}_1(\hat{\bm{\zeta}}))   
  +
  \EZvorTwoZeta ( 
  \norm{\hat{\bm{\zeta}}^{(2)}-\bm{\theta}}^2 ) \PZvor(\mathcal{V}_2(\hat{\bm{\zeta}})).
   \end{equation}
Because $\bm{\theta}$ is continuous in $\Theta$ and
$\hat{\bm{\zeta}}^{(2)}$ has a bounded norm,
$\PZvor(\mathcal{V}_2(\hat{\bm{\zeta}})) > 0$.  Thus, using the
definition \eqref{eq:defi:D2} of $\mathcal{V}_2$, \eqref{eq: app:
  AAAxi} yields
\begin{align}\label{eq:J_lt_JMS}
\JZ(\hat{\bm{\zeta}})
&  < \EZvorOneZeta ( 
\norm{\hat{\bm{\zeta}}^{(1)}-\bm{\theta}}^2 
) \PZvor(\mathcal{V}_1(\hat{\bm{\zeta}}))   \notag  \\
& \quad +
\EZvorTwoZeta ( 
\norm{\hat{\bm{\zeta}}^{(1)}-\bm{\theta}}^2 ) \PZvor(\mathcal{V}_2(\hat{\bm{\zeta}}))
=
\varphi(\hat{\bm{\zeta}}^{(1)})
=\JMS.
\end{align}

We next choose $b_0$ based on $\hat{\bm{\zeta}}$. Let
\begin{equation}\label{eq: Lemma: eps choice}
\epsilon \dfn \JMS - \JZ({\hat{\bm{\zeta}}}).
\end{equation}
According to \eqref{eq:J_lt_JMS} $\epsilon>0$, and Lemma~\ref{lem:
  Heter: \J>MMSE-eps} states that there exists a number
$a(\epsilon)\in\PositiveReals$ such that if
$\hat{\bm{\theta}} \not\in \Sset^{a(\epsilon)}$
then
$\JZ({\hat{\bm{\theta}}}) > \JMS-\epsilon =
\JZ({\hat{\bm{\zeta}}})$. Set
    \begin{equation}
    b_0 \dfn 1 + a(\epsilon) \vee \norm{\thetaMS}
    \end{equation}
    and consider the set $\Sset^b$ for any number $b \geq b_0$.
    Because $\Sset^b$ is closed and bounded, it is compact. Since
    $\JZ$ is continuous everywhere, it is necessarily continuous
    in $\Sset^b$. Thus, according to the Weierstrass extreme value
    theorem, $\JZ$ attains a minimum in at least one point in
    $\Sset^b$ (the cost at all such points is, of course, identical).
    Henceforth denoting any of these minimum points as
    $\bm{\hat{\bm{\psi}}}$, the first part of the lemma states that
    $\JZ(\bm{\hat{\bm{\psi}}}) < \JMS$.  To show this, we notice
    that
\begin{equation}
\norm{\hat{\bm{\zeta}}^{(1)}}
\vee
\norm{\hat{\bm{\zeta}}^{(2)}}
=
\norm{\thetaMS}  \vee c
 < b_0
\end{equation}
so that
$\hat{\bm{\zeta}} \in \interior{(\Sset^{b_0})}$,
where we use the notation $\interior{A}$ for the interior of $A$, and,
consequently,
$\hat{\bm{\zeta}} \in \interior{(\Sset^{b})}$.  Since
$\JZ(\hat{\bm{\psi}}) \leq \JZ(\hat{\bm{\zeta}})$, the first part of the lemma
follows upon invoking \eqref{eq:J_lt_JMS}.

Continuing to the second part of the lemma, we now prove, by
contradiction, that $\hat{\bm{\psi}}$ satisfies
(\ref{eq:Heter:eq_from_first_derivatives_prelim}).  If it does not,
then either 1) it is an interior point where $\JZ$ is
non-differentiable, or 2) it is located along the boundary of
$\Sset^b$. We contradict each case separately.

We first show that none of the minimum points can be located at an
interior point where $\JZ$ is non-differentiable.  Let
\begin{equation}\label{eq: Y defi}
  \mathcal{Y} \dfn
  \left\{
    \hat{\bm{\theta}} \in \augmentedSpaceRealization \mid
    \hat{\bm{\theta}}^{(1)} = \hat{\bm{\theta}}^{(2)}
  \right\}.
\end{equation}
Each point in $\mathcal{Y}$ is a pair of estimates for which $\JZ$ is
non-differentiable.  Assume that $\hat{\bm{\psi}} \in \interior{(\Sset^b)} \cap \mathcal{Y}$. Then $\hat{\bm{\psi}}^{(1)}=\hat{\bm{\psi}}^{(2)}$, giving
\begin{align}\label{eq: app: AAApsi}
\JZ(\hat{\bm{\psi}}) &=
\varphi(\hat{\bm{\psi}}^{(1)}).
\end{align}
Now set $\hat{\bm{\xi}}^{(1)} \dfn \hat{\bm{\psi}}^{(1)}$ and
$\hat{\bm{\xi}}^{(2)} \dfn c b \bm{e_{\theta_1}}$ with $c\in (0,1) $
such that $\hat{\bm{\xi}}^{(2)} \neq \hat{\bm{\psi}}^{(1)}$.
Obviously, $\hat{\bm{\xi}} \not\in \mathcal{Y}$, and
\begin{equation}
\norm{\hat{\bm{\xi}}^{(1)}}
\vee
\norm{\hat{\bm{\xi}}^{(2)}}
=
\norm{\hat{\bm{\psi}}^{(1)}} \vee cb < b
\end{equation}
so that $\hat{\bm{\xi}} \in \interior{(\Sset^b)}$.
Analogously to \eqref{eq:J_lt_JMS} and using \eqref{eq: app: AAApsi}
we have
\begin{equation}
\JZ(\hat{\bm{\xi}})
<
\varphi(\hat{\bm{\xi}}^{(1)})
=
\varphi(\hat{\bm{\psi}}^{(1)})
=
\JZ(\hat{\bm{\psi}}),
\end{equation}
contradicting the assumption that $\JZ$ has a minimum at
$\hat{\bm{\psi}}\in \interior{(\Sset^b)}$.

Now assume that $\hat{\bm{\psi}}$ is located along the boundary of
$\Sset^b$. Then $\hat{\bm{\psi}} \not\in \Sset^{a(\epsilon)}$, and,
according to Lemma~\ref{lem: Heter: \J>MMSE-eps},
$\JZ({\hat{\bm{\psi}}}) > \JZ({\hat{\bm{\zeta}}})$. Recalling that
$\hat{\bm{\zeta}} \in \interior{(\Sset^{b})}$ 
yields a contradiction to the assumption that $\hat{\bm{\psi}}$ is a minimum
point of $\JZ$ in
$\Sset^{b}$.
\end{proof}

Returning to the proof of Theorem~\ref{theo: Heter optimality}, let
$b_0$ be a number satisfying Lemma \ref{lem: Heter: min of \J in the
  ball}, and consider $b \geq b_0$ . Then, according to Lemma
\ref{lem: Heter: min of \J in the ball}, there exists at least one
minimizer of $\JZ$ in $\Sset^{b}$, at a point that satisfies
(\ref{eq:Heter:eq_from_first_derivatives_prelim}).  Denote the value
of $\JZ$ at any such minimum point as $\JZ^*$ (recall that
the costs at all minimum points are identical). Then, according to
Lemma \ref{lem: Heter: min of \J in the ball},
\begin{equation}\label{eq: Theorem: \J* < MMSE}
\JZ^* < \JMS.
\end{equation}
Define
\begin{equation}\label{eq: defi: eps2}
\epsilon \dfn \JMS-\JZ^*.
\end{equation}
According to Lemma \ref{lem: Heter: \J>MMSE-eps}, there exists a
number $a(\epsilon)>0$ such that
\begin{equation}
 \JZ({\hat{\bm{\theta}}}) > \JMS-\epsilon = \JZ^* \quad 
\forall 
\hat{\bm{\theta}} \not\in \Sset^{a(\epsilon)}.
\end{equation}

Now consider the set $\Sset^{c}$ with
\begin{equation}\label{eq: Theorem: a choice}
c \dfn a(\epsilon) \vee b.
\end{equation}
According to Lemma \ref{lem: Heter: min of \J in the ball} and based
on (\ref{eq: Theorem: a choice}), there is {at least one} minimum
point in $\Sset^{c}$, where the cost {(which is identical for all
  minimum points)} does not exceed $\JZ^*$ because
$\Sset^{a(\epsilon)} \in \Sset^{c}$.  At {any such} point, thus,
the cost is strictly smaller than $\JMS$ due to (\ref{eq: Theorem: \J*
  < MMSE}). Outside of $\Sset^{c}$, $\JZ>\JZ^*$ due to
Lemma \ref{lem: Heter: \J>MMSE-eps} and (\ref{eq: Theorem: a choice}).
\end{proof}
%%
%%
%%
%%% Local Variables:
%%% mode: latex
%%% TeX-master: "main"
%%% End:

\section{Proof of Theorem~\ref{theo: Hier optimality}}\label{sec:app_proof_opt_hier}
\begin{proof}
  Let $\mathcal{T} \subset \augmentedSpaceRealization $ be defined as
  $\mathcal{T} \dfn \{ \hat{\bm{\theta}}^{(1)} = \thetaMS \} \times
  \EstimSpaceRealization$.
  In the hierarchical problem \eqref{eq:Hier:Definition} the objective
  is to minimize a restriction of the cost $\JZ$, originally defined
  in \eqref{eq: definition of \JZ} on
  $\augmentedSpaceRealization$, to the subdomain $\mathcal{T}$, that
  is
\begin{equation}\label{eq: JZrest definition}
  \JZrest (\hat{\bm{\theta}}^{(2)})
  \dfn \JZ ( \hat{\bm{\theta}}^{(1)}
  = \thetaMS, \hat{\bm{\theta}}^{(2)}).
\end{equation}
As a restriction of $\JZ$, $\JZrest$ is continuous everywhere and differentiable except at
$\bm{\hat{\bm{\theta}}}^{(2)} = \thetaMS$.

The proof of Theorem~\ref{theo: Hier optimality} follows the proof of
Theorem~\ref{theo: Heter optimality} in a restricted form. We begin
with the following lemma.
\begin{lemma}\label{lem: Hier: \J>MMSE-eps}
  For any number $\epsilon\in\PositiveReals$ there exists a number
  $a(\epsilon)\in\PositiveReals$ such that if
  $\norm{\hat{\bm{\theta}}^{(2)}} > a(\epsilon)$ then
  $\JZrest({\hat{\bm{\theta}}^{(2)}}) > \JMS-\epsilon$.
\end{lemma}
\begin{proof}
  Because $\JZrest$ is a restriction of $\JZ$, Proposition~\ref{prop:D}
  applies with $\hat{\bm{\theta}}^{(1)} = \thetaMS$; hence, there exists a
  number $L(\epsilon)\in\PositiveReals$ such that if
  $\mathcal{D}>L(\epsilon)$ then $\JZrest > \JMS - \epsilon$.  Set
\begin{equation}\label{eq: Hier Lemma: a choice}
a(\epsilon) \dfn 1 + L(\epsilon) + \norm{\thetaMS}
\end{equation}
and assume that $\norm{\hat{\bm{\theta}}^{(2)}} > a(\epsilon)$.  The
lemma follows upon observing that
\begin{equation}
\mathcal{D}  \geq  1 + L(\epsilon) > L(\epsilon).
\end{equation}
\end{proof}
\begin{lemma}\label{lem: Hier: min of \J in the ball}
  There exists a number $b_0\in\PositiveReals$ such that for any
  $b \geq b_0$, $\JZrest$ attains a minimum in
  $\{ \hat{\bm{\theta}}^{(2)}\in\EstimSpaceRealization \mid
  \norm{\hat{\bm{\theta}}^{(2)}} \leq b \}$,
  at a point where $\JZrest < \JMS$. This point satisfies
  \eqref{eq:Hier:eq_from_first_derivatives_prelim}.
\end{lemma}
\begin{proof}
  The proof follows the proof of Lemma \ref{lem: Heter: min of \J in the ball} in restricted form, where $\JZ$ is replaced by $\JZrest$,
  $\Sset^b$ by its subset
  $\Sset^b \cap \{
  \norm{\hat{\bm{\theta}}^{(1)}} = \norm{\thetaMS} \}$,
  Lemma \ref{lem: Heter: \J>MMSE-eps} by Lemma \ref{lem: Hier:
    \J>MMSE-eps}, and
  (\ref{eq:Heter:eq_from_first_derivatives_prelim}) by 
  (\ref{eq:Hier:eq_from_first_derivatives_prelim}).
\end{proof}
Having these two lemmas on hand, the proof of Theorem~\ref{theo: Hier
  optimality} follows the proof of Theorem~\ref{theo: Heter
  optimality} in a restricted form, where $\JZ$ is replaced by $\JZrest$,
the domain $\augmentedSpaceRealization$ by the subdomain
$\mathcal{T}$, Lemma \ref{lem: Heter: \J>MMSE-eps} by Lemma \ref{lem:
  Hier: \J>MMSE-eps}, Lemma \ref{lem: Heter: min of \J in the ball} by
Lemma \ref{lem: Hier: min of \J in the ball}, and
(\ref{eq:Heter:eq_from_first_derivatives_prelim}) by
(\ref{eq:Hier:eq_from_first_derivatives_prelim}).
\end{proof}

%%% Local Variables:
%%% mode: latex
%%% TeX-master: "main"
%%% End:

\section{The Gaussian Altruism Equations}\label{sec:altruistic_eq}
This Appendix investigates equations (\ref{eq:Heter:altruism
  equation}) and (\ref{eq:Hier:altruism equation}), the heterarchical
and hierarchical altruism equations, respectively, in the
\emph{Gaussian} case. Recall that both equations are scalar, algebraic
equations, that depend on $\hat{u}_m$.

In the sequel we define $x$ to be the realization of the random
variable $\hat{u}^m$ corresponding to the realization $z$ of the
measurement vector $Z$. Furthermore, given the realization $z$, we
define the conditional random vector $Y$ as
\begin{equation}
  \label{eq:Y_def}
  Y \dfn u_1 \mid Z=z
\end{equation}
and assume that $Y\sim\gaus{\mu_Y}{\sigma_Y^2}$.
\subsection{Heterarchical Altruistic Estimation}\label{Appendix:Heter}
For a given realization $z$ of $Z$, and using the definition
\eqref{eq:Y_def}, the heterarchical altruism equation
(\ref{eq:Heter:altruism equation}) reduces to
\begin{equation} \label{eq:fundamental}
  x = \frac{1}{2} [ \expect(Y\mid Y<x) + \expect(Y\mid Y>x) ].
\end{equation}
Let $X$ be the standardized version of $Y$
\begin{equation}\label{eq:chi_denotation}
X \dfn \frac{Y - \mu_Y}{\sigma_Y}
\end{equation}
with probability density function and cumulative density function
$\phi$ and $\Phi$, respectively. Then, letting $\chi$ be the
realization of $X$ corresponding to the realization $x$ of $Y$, we
have~\cite[Theorem 19.2]{Greene:2012}
\begin{subequations}\label{eq:rohatgi}
  \begin{align}
    \expect(Y\mid Y > x) & = 
  \mu_Y + \sigma_Y\frac{\phi(\chi)}{1 - \Phi(\chi)}\label{eq: Appendix: Eyx}\\
    \expect(Y\mid Y < x) & = \mu_Y - \sigma_Y\frac{\phi(\chi)}{\Phi(\chi)}.
\end{align}
\end{subequations}
Using equations \eqref{eq:rohatgi} in \eqref{eq:fundamental} yields
\begin{equation}\label{eq:Heter:chi_eq}
\frac{\phi(\chi)}{2[1-\Phi(\chi)]}
-
\frac{\phi(\chi)}{2\Phi(\chi)}
- \chi = 0.
\end{equation}

It is easy to see that $\chi=0$ is a solution of
(\ref{eq:Heter:chi_eq}).  Noting \eqref{eq:chi_denotation}, this
solution of the realization-based \eqref{eq:Heter:chi_eq} is
equivalent to \eqref{eq:Heter:um_is_mu1}, the solution of the general
equation \eqref{eq:Heter:altruism equation}, thus proving
Proposition~\ref{prop:Heter altruistic}.  In the following lemma we
prove that this solution is unique.
\begin{lemma}
  $\chi=0$ is the only solution of (\ref{eq:Heter:chi_eq}).
\end{lemma}
\begin{proof}
  To prove the lemma we define the function
\begin{equation}
f_\Heter(\chi)
\dfn
\frac{\phi(\chi)}{2[1-\Phi(\chi)]}
-
\frac{\phi(\chi)}{2\Phi(\chi)}
- \chi
\end{equation}
and show that $\chi=0$ is its only zero.  Since $f_\Heter$ is an
anti-symmetric function, it is sufficient to prove that it does not
vanish in $\left(0,\infty\right)$. We do this by proving that
$f_\Heter(\chi)<0$ for all $\chi>0$.

Defining
\begin{align}
f^{(1)}_\Heter(\chi) &\dfn -\frac{\phi(\chi)}{2\Phi(\chi)}
+ \phi(0) - \frac{\chi}{2}
\\
\intertext{and}
f^{(2)}_\Heter(\chi) &\dfn \frac{\phi(\chi)}{2[1-\Phi(\chi)]}
- \phi(0) - \frac{\chi}{2}
\end{align}
yields
\begin{equation}
f_\Heter(\chi) = f^{(1)}_\Heter(\chi) + f^{(2)}_\Heter(\chi).
\end{equation}
We thus proceed to prove, separately, that both
\begin{equation}\label{eq:first_ineq}
f^{(1)}_\Heter(\chi) < 0 \quad \forall \chi>0
\end{equation}
and
\begin{equation}\label{eq:second_ineq}
f^{(2)}_\Heter(\chi) < 0 \quad \forall \chi>0.
\end{equation}

To prove \eqref{eq:first_ineq} we recast it as
\begin{equation}
  g^{(1)}_\Heter(\chi) < 0
  \quad
  \forall \chi>0
\end{equation}
where
\begin{equation}\label{eq:g_1_HT}
  g^{(1)}_\Heter(\chi) \dfn
  -\phi(\chi)
  + 2\phi(0)\Phi(\chi) - \chi\Phi(\chi).
\end{equation}
The definition \eqref{eq:g_1_HT} gives
\begin{equation}
g^{(1)}_\Heter(0) = 0.
\end{equation}
Since $g^{(1)}_\Heter(\chi)$ is a continuous function of its argument
over $\left(0,\infty\right)$, it suffices to show that it is
monotonically strictly decreasing in that interval. To this end, we
calculate its derivative
\begin{equation}
  g'^{(1)}_\Heter(\chi)
  =
  2\phi(0)\phi(\chi) - \Phi(\chi)
\end{equation}
and notice that, since $\phi(\chi)<\phi(0)$ and
$\Phi(\chi)>\frac{1}{2}$ in $\left(0,\infty\right)$,
\begin{equation}
  g'^{(1)}_\Heter(\chi)
  <
  2\phi(0)^2 - \frac{1}{2}
  \approx
  -0.1817  \quad \forall \chi >0.
\end{equation}
For illustrative purposes, the function $g^{(1)}_\Heter(\chi)$ and its
first derivative are depicted in Fig.~\ref{fig:Heter_chi_eq_g1}.
\begin{figure}[tbhp]
\centering
\includegraphics[width=\linewidth]{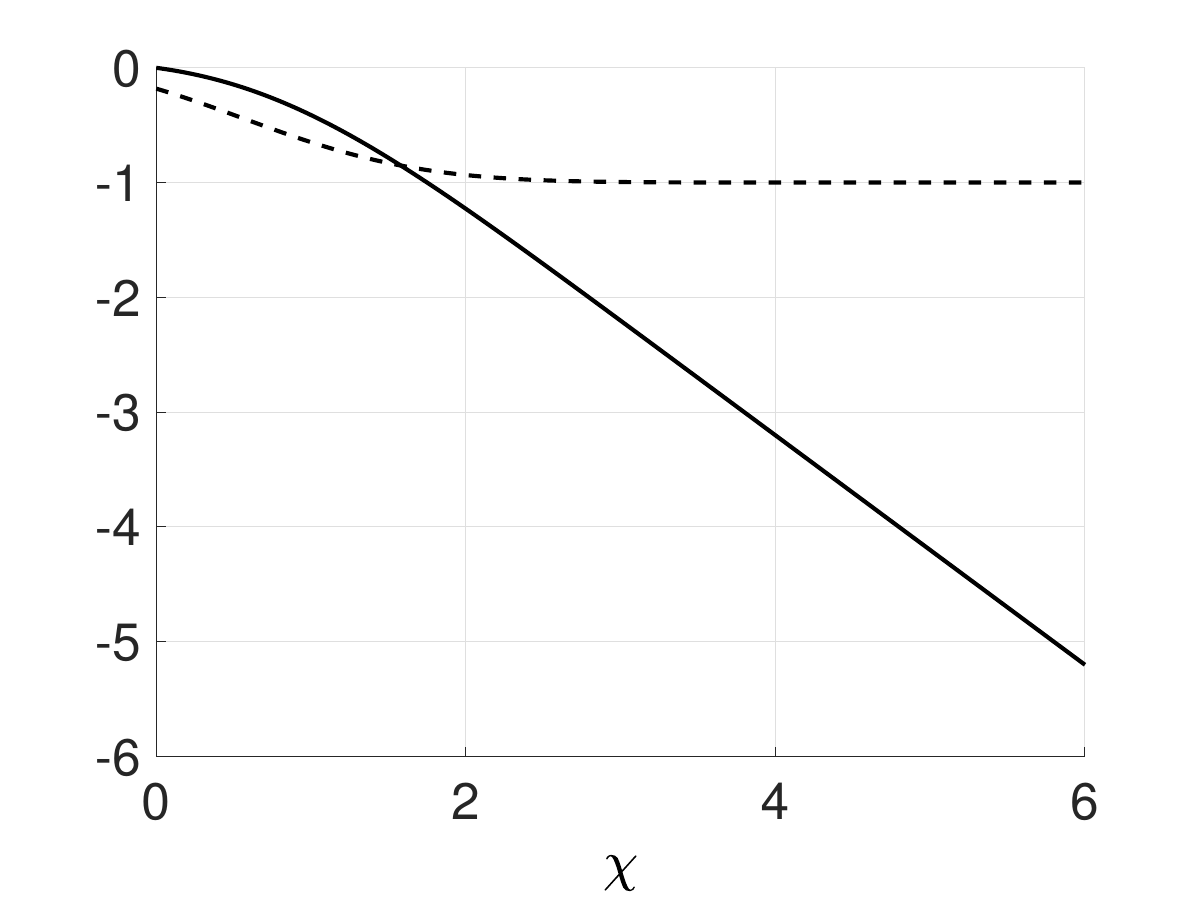}
\caption{The function $g^{(1)}_\Heter(\chi)$ (solid line) and its
  first derivative (dashed line).}
\label{fig:Heter_chi_eq_g1}
\end{figure}

To prove \eqref{eq:second_ineq}, we recast it as
\begin{equation}
\label{eq:def_g2}
  g^{(2)}_\Heter(\chi) < 0 \quad \forall \chi>0
\end{equation}
where
\begin{equation}
  g^{(2)}_\Heter(\chi) \dfn \phi(\chi) - 2\phi(0) + 2\phi(0)\Phi(\chi)
  - \chi[1-\Phi(\chi)].
\end{equation}
The proof rests on the continuity of $g^{(2)}_\Heter(\chi)$ and its
derivatives over $\left[0,\infty\right)$. Calculating the first three
derivatives of $g^{(2)}_\Heter(\chi)$ yields
\begin{align}
g'^{(2)}_\Heter(\chi)  & = \Phi(\chi) + 2\phi(0)\phi(\chi) - 1
\\
g''^{(2)}_\Heter(\chi) & = \phi(\chi)[1-2\phi(0)\chi]
\\
g'''^{(2)}_\Heter(\chi) & = \phi(\chi)[2\phi(0)(\chi^2-1)-\chi].
\end{align}
Clearly, $\chi_2\dfn\frac{1}{2\phi(0)}$ is the single zero of
$g''^{(2)}_\Heter(\chi)$ in its entire domain. Since
$g'''^{(2)}_\Heter(\chi_2) = -2\phi(0)\phi(\chi_2) < 0$, $\chi_2$ is
the single maximum point of $g'^{(2)}_\Heter(\chi)$, and we compute
\begin{equation}
g'^{(2)}_\Heter(\chi_2) \approx 0.04.
\end{equation}
We now investigate the behavior of $g'^{(2)}_\Heter(\chi)$ on both
sides of its single extremal point. Clearly,
$g'^{(2)}_\Heter(\chi) \to 0$ as $\chi\to\infty$.  Since $\chi_2$ is
the only extremal point of $g'^{(2)}_\Heter(\chi)$, this yields that
$g'^{(2)}_\Heter(\chi) > 0$ for all $\chi \geq \chi_2$, from which we
conclude that $g^{(2)}_\Heter(\chi)$ does not have any extremal point
in $\left[ \chi_2, \infty \right)$.

Turning our attention to $\chi < \chi_2$, we first notice that
$g'^{(2)}_\Heter(0) \approx -0.1817$.  Since
$g'^{(2)}_\Heter(\chi_2) >0$, the mean value theorem yields that
$g'^{(2)}_\Heter(\chi)$ must have a zero in $\left(0, \chi_2\right)$.
Moreover, since $\chi_2$ is the only extremal point of
$g'^{(2)}_\Heter(\chi)$, we conclude that $g'^{(2)}_\Heter(\chi)$ is
monotonically strictly increasing in $\left(0, \chi_2\right)$, so that its
zero in $\left(0, \chi_2\right)$ is unique. Denote that zero as
$\chi_1$. Clearly, $\chi_1$ is a unique minimum point of
$g^{(2)}_\Heter(\chi)$ in $\left[ 0, \chi_2\right]$, since
$g'^{(2)}_\Heter(\chi)<0$ for $0\leq\chi<\chi_1$ and
$g'^{(2)}_\Heter(\chi)>0$ for $\chi_1<\chi\leq\chi_2$. Noting that
$g^{(2)}_\Heter(0) = 0$ and $g^{(2)}_\Heter(\chi_2) \approx -0.0336$,
we conclude that $g^{(2)}_\Heter(\chi) < 0$ in
$(0,\chi_2]$.

To prove that $g^{(2)}_\Heter(\chi) < 0$ also for $\chi>\chi_2$, we
observe that $g^{(2)}_\Heter(\chi)\to 0$ as $\chi\to\infty$, which
follows from both
\begin{equation}
  \phi(\chi) - 2\phi(0) + 2\phi(0)\Phi(\chi)\to 0 
  \quad \text{as } \chi\to\infty
\end{equation}
and
\begin{equation}\label{eq:lhopital}
  \chi[1-\Phi(\chi)] \to 0 \quad \text{as } \chi\to\infty 
\end{equation}
where the latter limit results from using L'H\^{o}pital's rule. Since
$g^{(2)}_\Heter(\chi_2)<0$, the proof then follows from our previous
conclusion that $g^{(2)}_\Heter(\chi)$ does not have any extremal
point in $\left[\chi_2,\infty\right)$. This also concludes the proof
of the lemma.

For illustrative purposes, the function $g^{(2)}_\Heter(\chi)$ and its
first two derivatives are depicted in Fig.~\ref{fig:Heter_chi_eq_g2}.
\begin{figure}[tbhp]
\centering
\includegraphics[width=\linewidth]{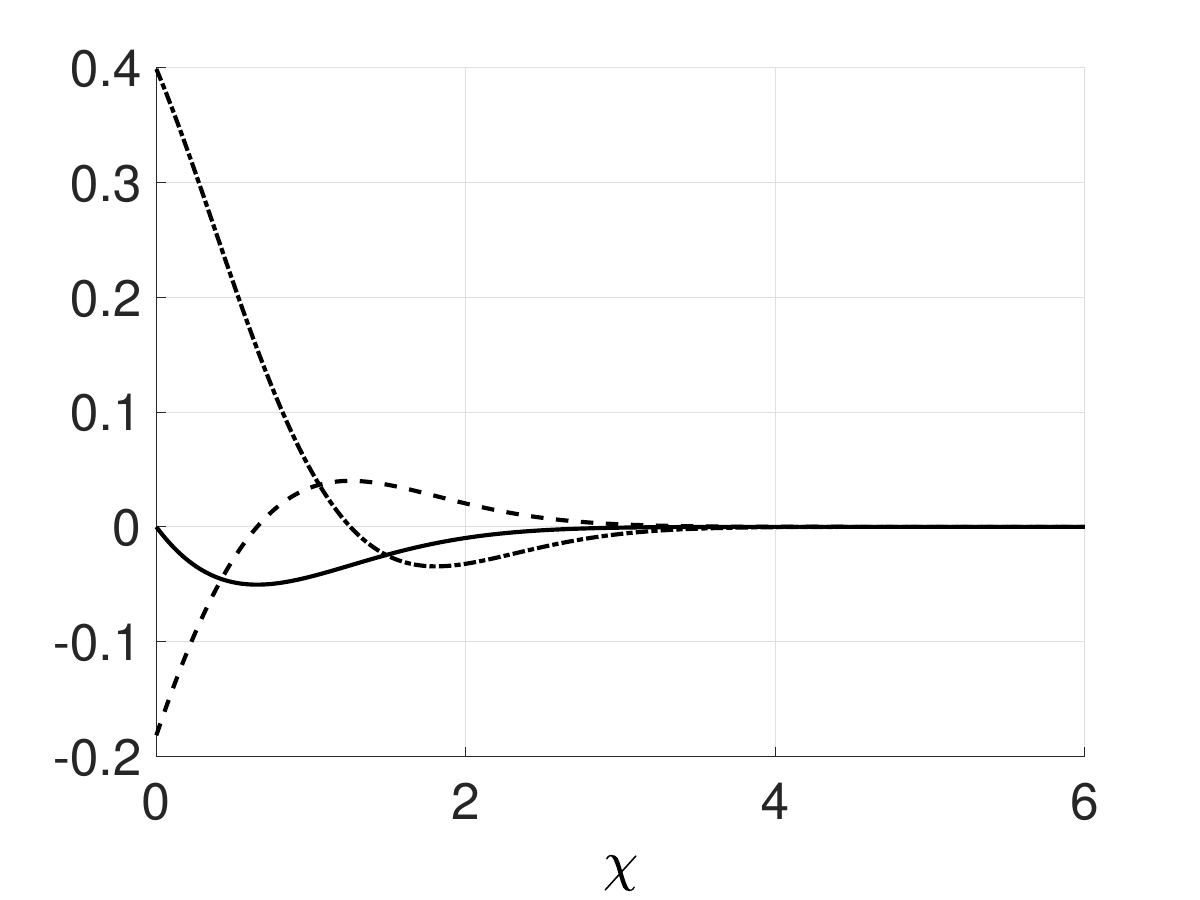}
\caption{The function $g^{(2)}_\Heter(\chi)$ (solid line), its first
  derivative (dashed line), and its second derivative (dash-dotted
  line).}
\label{fig:Heter_chi_eq_g2}
\end{figure}
\end{proof}
\subsection{Hierarchical Altruistic Estimation}\label{Appendix:Hier}
For a given realization $z$ of $Z$ the hierarchical altruism equation
(\ref{eq:Hier:altruism equation}) reduces to
\begin{equation}\label{eq: Appendix: hier tmp}
  x = \frac{1}{2} [\mu_Y + \expect(Y\mid Y>x) ]
\end{equation}
which, using (\ref{eq: Appendix: Eyx}), yields
\begin{equation}\label{eq:Hier:chi_eq}
  \frac{\phi(\chi)}{2[1-\Phi(\chi)]}
  - \chi = 0.
\end{equation}

As opposed to its heterarchical counterpart \eqref{eq:Heter:chi_eq},
the hierarchical equation \eqref{eq:Hier:chi_eq} does not lend itself
to a closed form solution. Resorting to numerical methods we find its
solution to be
\begin{subequations}
		\label{eq:hier_soln}
\begin{gather}
  \chi = \frac{1}{2}\wHier \label{eq:138a}
\intertext{with}
	\wHier \approx 1.224 \label{eq:wHier_value_app}.
\end{gather}
\end{subequations}
Furthermore, this solution is unique, as proved in the following
lemma.
Noting \eqref{eq:chi_denotation}, this solution of the
realization-based \eqref{eq:Hier:chi_eq} is equivalent to (\ref{eq:
  Hier estimators}), the solution of the general equation
\eqref{eq:Hier:altruism equation}.
\begin{lemma}\label{lem:chi_eq_Soln}
  \eqref{eq:hier_soln} is the unique solution to
  \eqref{eq:Hier:chi_eq}.
\end{lemma}
\begin{proof}
  Clearly, (\ref{eq:Hier:chi_eq}) cannot have a non-positive
  solution. To prove that it cannot have a positive solution other
  than \eqref{eq:hier_soln}, we recast (\ref{eq:Hier:chi_eq}) as
\begin{equation}
g_\Hier(\chi) = 0
\end{equation}
where
\begin{equation}
g_\Hier(\chi) \dfn \phi(\chi) - 2\chi[1-\Phi(\chi)].
\end{equation}
The rest of the proof relies on the continuity of $g_\Hier(\chi)$ and
its derivatives, the first three of which are calculated to be
\begin{align}
  g'_\Hier(\chi) &= \chi\phi(\chi) - 2[1-\Phi(\chi)]
                   \label{eq:hier_1st_der}\\
  g''_\Hier(\chi) &= \phi(\chi)(3-\chi^2)
                    \label{eq:hier_2nd_der}\\
  g'''_\Hier(\chi) & = \chi\phi(\chi^2 - 5)
                     \label{eq:hier_3rd_der}.
\end{align}
To facilitate the ensuing development, we summarize in Table~\ref{tab:
  chi hier} the signs of $g_\Hier(\chi)$ and its first three
derivatives at $\chi=0$ and at $\chi=\sqrt{3}$.
\begin{table}[tbhp]
  \caption{Signs of $g_\Hier(\chi)$ and its first three derivatives at $\chi=0$ and $\chi=\sqrt{3}$.}
\begin{center}
\begin{tabular}{ccc}
  \hline\hline
  & $\chi=0$ & $\chi=\sqrt{3}$ \\ 
  \hline
  $\sign g_\Hier(\chi)$  & $+$ & $-$ \\ 
  $\sign g'_\Hier(\chi)$  & $-$ & $+$ \\ 
  $\sign g''_\Hier(\chi)$  & $+$ &  0\\
  $\sign g'''_\Hier(\chi)$  & 0  &  $-$\\
  \hline\hline
\end{tabular}
\end{center}
\label{tab: chi hier}
\end{table}
The only root of $g''_\Hier(\chi)$ in $[0,\infty)$ is $\chi=\sqrt{3}$.
Since $g'''_\Hier(\sqrt{3}) < 0$, it follows that this root is the
only maximum point of $g'_\Hier(\chi)$ in $[0,\infty)$, rendering
$g'_\Hier(\chi)$ monotonically non-increasing for $\chi > \sqrt{3}$.
Furthermore, since $g'_\Hier(\sqrt{3}) > 0$ and $g'_\Hier(\chi)\to 0$
as $\chi\to\infty$, it follows that $g'_\Hier(\chi) > 0$ in
$[\sqrt{3},\infty)$, rendering $g_\Hier(\chi)$ monotonically strictly
increasing in that interval.  Now, using \eqref{eq:lhopital},
it is easy to see that $g_\Hier(\chi)\to 0$ as $\chi\to\infty$.  Since
$g_\Hier(\sqrt{3}) < 0$, we conclude that $g_\Hier(\chi)$ does not
possess any root in $[\sqrt{3},\infty)$.

To complete the proof, we need to show that $g_\Hier(\chi)$ does not
possess any root in $(0,\sqrt{3})$ other than \eqref{eq:hier_soln}
(as Table~\ref{tab: chi hier} shows that both 0 and $\sqrt{3}$ are not
roots of $g_\Hier(\chi)$). Since $g'_\Hier(0)<0$ and
$g'_\Hier(\sqrt{3})>0$, and as $g'_\Hier(\chi)$ does not possess an
extremum in $(0,\sqrt{3})$, it must be monotonically increasing in
that interval, crossing zero at a single point in $(0,\sqrt{3})$. Thus,
$g_\Hier(\chi)$ can have only a single extremal point in that
interval. This extremal point is a minimum point since
$g''_\Hier(\chi)>0$ in $(0,\sqrt{3})$. Since $g_\Hier(0)>0$ and
$g_\Hier(\sqrt{3})<0$, we thus conclude that $g_\Hier(\chi)$ can cross
zero only once in $(0,\sqrt{3})$. This unique crossing is at
\eqref{eq:hier_soln}.

For illustrative purposes, the function $g_\Hier(\chi)$ is depicted in Fig.~\ref{fig:Hier_chi_eq}.
\begin{figure}[tbhp]
\centering
\includegraphics[width = \linewidth]{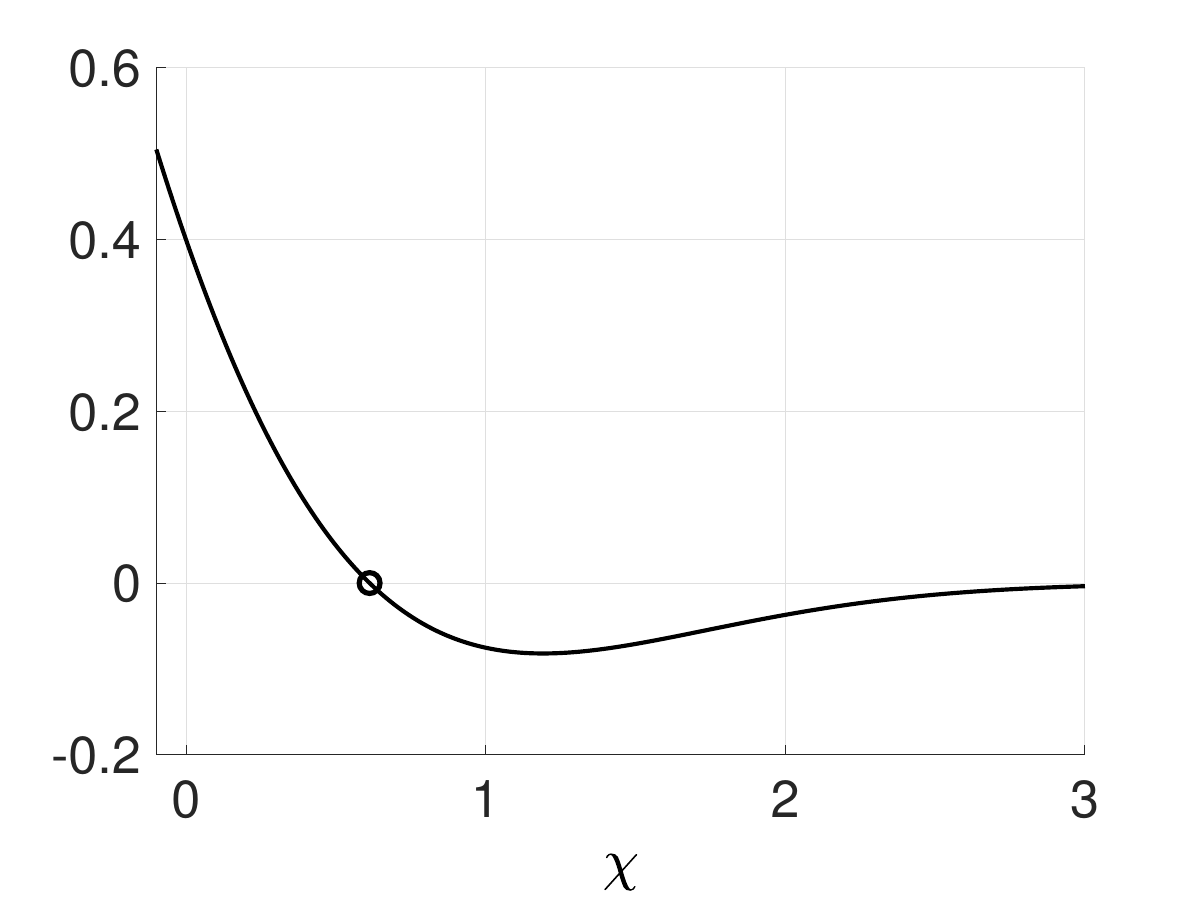}
\caption{The function $g_\Hier(\chi)$. Superimposed is its numerically
  calculated root, $\chi\approx 0.612$ (circle).}
\label{fig:Hier_chi_eq}
\end{figure}
\end{proof}

\end{appendices}

\section*{Acknowledgment}
The authors thank Vadim Indelman of the Technion's Department of
Aerospace Engineering for his useful suggestions.

\bibliographystyle{IEEEtran}
\bibliography{bib}

\end{document}